\documentclass[acmsmall]{acmart}

\usepackage[utf8]{inputenc}
\usepackage{amsmath}
\usepackage{ebproof} %for displaying formal proofs
\usepackage{amsthm}
\theoremstyle{plain}
\newtheorem{theorem}{Theorem}
\newtheorem{lemma}[theorem]{Lemma}
\newtheorem{corollary}[theorem]{Corollary}

\theoremstyle{definition}
\newtheorem{definition}[theorem]{Definition}
\newtheorem{example}[theorem]{Example}
\newtheorem{remark}[theorem]{Remark}

\usepackage{fixme}
\fxsetup{
	status=draft,
	author=,
	layout=footnote,
	theme=color
}

\newcommand{\calA}{\mathcal{A}}
\newcommand{\calB}{\mathcal{B}}
\newcommand{\calL}{\mathcal{L}}
\newcommand{\calM}{\mathcal{M}}
\newcommand{\calC}{\mathcal{C}}
\newcommand{\calV}{\mathcal{V}}
\newcommand{\calT}{\mathcal{T}}

\renewcommand{\phi}{\varphi}
\newcommand{\mybar}[1]{\overline{#1}}
\newcommand{\Nat}{\mathbb{N}}
\newcommand{\Int}{\mathbb{Z}}
\newcommand{\Rat}{\mathbb{Q}}
\newcommand{\id}{{\mathrm{id}}}
\newcommand{\powset}{\mathcal{P}}
\newcommand{\union}{\cup}
\newcommand{\sop}{[} % open substitution syntax
\newcommand{\scl}{]} % close substitution syntax
\newcommand{\sel}[2]{#1 \backslash #2} % element of substitution syntax
\newcommand{\unsubst}[2]{\sop \sel{#1}{#2} \scl} % unary substitution
\newcommand{\impl}{\rightarrow} % logical implication
\newcommand{\liff}{\leftrightarrow} % logical iff
\newcommand{\Land}{\bigwedge}
\newcommand{\Lor}{\bigvee}
\newcommand{\proves}{\vdash}
\newcommand{\lfp}{\mathrm{lfp}} % least-fixed point

\newcommand{\hoare}[3]{\{#1\}#2\{#3\}}
\newcommand{\vc}{\mathrm{vc}}
\newcommand{\vct}{\tilde{\vc}}
\newcommand{\hskipit}{\textbf{skip}}
\newcommand{\hif}{\textbf{if}~}
\newcommand{\hthen}{~\textbf{then}~}
\newcommand{\helse}{~\textbf{else}~}
\newcommand{\hwhile}{\textbf{while}~}
\newcommand{\hdo}{~\textbf{do}~}
\renewcommand{\wp}{\mathrm{wp}}
\newcommand{\Rwp}{R_{\wp}}
\renewcommand{\sp}{\mathrm{sp}}
\newcommand{\wpc}{\mathrm{wp}}
\newcommand{\spc}{\mathrm{sp}}

\newcommand{\Laff}{\calL_{\mathrm{aff}}}

\newcommand{\Aff}{{\mathrm{Aff}~}}
\newcommand{\aff}{{\mathrm{aff}}}

\newcommand{\modelsa}{\models_\mathrm{a}} % models abstractly
\newcommand{\nmodelsa}{\not\models_\mathrm{a}} % does not model abstractly

\begin{document}
%% title
\title{An abstract fixed-point theorem for Horn formula equations}

\author{Stefan Hetzl}
%\authornote{Both authors contributed equally to this research.}
\email{stefan.hetzl@tuwien.ac.at}
%\orcid{1234-5678-9012}
\affiliation{%
	\institution{Institute of Discrete Mathematics and Geometry, TU Wien}
	%\streetaddress{P.O. Box 1212}
	\city{Vienna}
	\country{Austria}
}

\author{Johannes Kloibhofer}
%\authornotemark[1]
\email{j.kloibhofer@uva.nl}
\affiliation{%
	\institution{Institute for Logic, Language and Computation, University of Amsterdam}
	%\streetaddress{P.O. Box 1212}
	\city{Amsterdam}
	\country{The Netherlands}
}

\renewcommand{\shortauthors}{S.~Hetzl \& J.~Kloibhofer}

\begin{abstract}
We consider a class of formula equations in first-order logic, Horn formula equations,
 which are defined by a syntactic restriction on the occurrences of predicate variables.
Horn formula equations play an important role in many applications in computer science.
We state and prove a fixed-point theorem for Horn formula equations in first-order logic with a
 least fixed-point operator.
Our fixed-point theorem is abstract in the sense that it applies to an abstract semantics which generalises
 standard semantics.
We describe several corollaries of this fixed-point theorem in various areas of computational logic,
 ranging from the logical foundations of program verification to inductive theorem proving.
\end{abstract}


\maketitle

\section{Introduction}

% brief motivation of problem from Ackermann to CHC solving
Solving Boolean equations is one of the oldest and one of the most central
 problems of logic.
It goes back to the 19th century and was already thoroughly investigated
 in~\cite{Schroeder90Vorlesungen}.
It has connections and applications throughout computational logic:
It can be understood as a variant of unification, see, e.g.,~\cite{Martin89Boolean}.
It is strongly related to second-order quantifier elimination and thus has
 many applications in areas such as modal logic, knowledge representation, and common-sense
 reasoning, see, e.g.,~\cite{Gabbay08Second}.
It is suitable as a logical foundation for software verification via constrained Horn
 clause solving, see, e.g.,~\cite{Bjorner15Horn}

% example (\models_{\calC}mula equation)
We are interested in formula equations in first-order logic.
A formula equation is simply an existential second-order formula,
 such as, for example,
\begin{equation}\label{eq.runex}
\exists X\ ( X(2) \land \forall u\, (X(u) \impl X(u + u)) \land \neg X(3) ). \tag{*}
\end{equation}
Solving this equation means finding a first-order formula $\phi(v)$
 such that, when $X(v)$ is replaced by $\phi(v)$, a valid
 first-order formula is obtained.
In this example the formula $\exists w\, w+w = v$, expressing that
 $v$ is even, would be a solution (assuming some basic background theory of natural
 number arithmetic).

% Horn formula equations, fixed points
In this paper we consider a certain class of formula equations, Horn
 formula equations.
Horn formula equations are characterised by the syntactic restriction of
 being a clause set where every clause contains at most one positive
 occurrence of a predicate variable (just as the above example (*) does).
In the context of verification, Horn formula equations are known as constrained Horn clause sets.
It is known that solving a Horn formula equation is strongly related
 to the computation of a least fixed point.
In our example (*), the definite clauses, i.e., those with at least
 one positive occurrence of the predicate variable $X$, define the
 mapping $F:\powset(\Nat) \to \powset(\Nat)$,
 $S \mapsto \{ 2 \} \union \{ 2\cdot n \mid n \in S \}$.
The least fixed point of $F$ is $P = \{ 2^m \mid m \geq 1 \}$.
By inserting $P$ for $X$, the definite clauses are satisfied trivially.
Thus (*) is seen to be equivalent to the conjunction
 of the remaining clauses, i.e., those that do not contain a positive
 occurrence of $X$, satisfying $P$.
In our case this amounts to the statement $3 \notin P$.
This is true, so (*) is true.

This solution $P$, the set of powers of two, is known to be undefinable in
the original language of the problem: that of linear
 arithmetic~\cite{Enderton01Mathematical}.
Another solution of (*) is  $E = \{ n \geq 2 \mid \text{$n$ is even} \}$
 which is definable in linear arithmetic (by the formula $\exists y\, y + y = x$).
The question of the expressivity of solutions in a restricted language is central for
 solving formula equations, and, in a wider context, for inductive theorem proving
 and software verification where the solutions of formula equations correspond to
 induction and loop invariants respectively.
The technique of abstract interpretation, introduced in~\cite{Cousot77Abstract}, is a flexible
 and powerful approach to dealing with such situations in software verification.
In this paper we introduce model abstractions, a notion of abstract Tarski model which
 corresponds to abstract interpretation and thus allows to deal with this phenomenon
 in the context of formula equations.

% our contribution: fixed point theorem
In this paper we formulate and prove a fixed point theorem for
 Horn formula equations which, in its simplest form, states the following:
Every Horn formula equations $\exists X\, \varphi(X)$ has a solution of the
 form $\lfp_{X}~\Phi(X)$ and, moreover, this is the least solution.
Here, $\Phi(X)$ is a formula that defines an operator and is induced by $\varphi(X)$.
The mere fact that this is true is well-known in constrained Horn clause solving
 and in logic programming.
Our main contributions are: 1.~We formulate and prove this result in a logic
 with an explicit fixed
 point operator in a general way that encompasses both, simultaneous least fixed points
 and  abstract interpretation.
2.~We show that this formulation with an explicit least fixed point operator is useful
 for a wide range of different applications.
In fact, about half of this paper is devoted to discussing these different applications:
as a simple corollary one can obtain the expressibility of the
 weakest precondition and the strongest postcondition, and thus the partial correctness of an imperative
 program in FO[LFP] without expressivity hypothesis.
It permits to considerably simplify the proof of the decidability of affine formula equations~\cite{Hetzl20Decidability}.
As another corollary it allows a generalisation of a result by
 Ackermann~\cite{Ackermann35Untersuchungen} on second-order quantifier-elimination in a direction different
 from the recent generalisation~\cite{Wernhard17Approximating}.
A result from a recently introduced approach to inductive theorem proving with tree grammars described
 in~\cite{Eberhard15Inductive} can also be obtained from our fixed-point theorem as another straightforward corollary.

This paper is structured as follows: in Section~\ref{sec.preliminaries} we introduce
 basic technical notions and results.
In Section~\ref{sec.form_eq} we introduce and discuss formula equations and, in particular, Horn
 formula equations.
The abstract fixed-point theorem is stated and proved in detail in Section~\ref{sec.afpthm}.
In Section~\ref{sectionAffineSolutionProblem} we obtain the decidability of the affine solution
 problem as corollary of the abstract fixed-point theorem.
Section~\ref{sec.app_prog_verif} describes applications of our fixed-point theorem to the logical foundations of program verification,
 in particular the definability of the weakest precondition and the strongest postcondition in FO[LFP] without
 expressivity hypothesis.
In Section~\ref{sec.approximation} we use it to approximate certain second-order formulas
 and in Section~\ref{sec.IndProving} we obtain a result about a method for inductive theorem proving
 based on tree grammars~\cite{Eberhard15Inductive} as corollary of the fixed-point theorem.
In Section~\ref{sec.related_work} we discuss related work.

Many of the results of this paper originate from the second author's master's thesis~\cite{Kloibhofer20Fixed}.
A concise presentation of the main results of the master thesis, in particular the fixed-point theorem for
 standard semantics, can be found in~\cite{Hetzl21Fixpoint}.
The abstract version of the fixed-point theorem and its corollaries have been obtained later and have been announced in the abstract~\cite{Hetzl21Abstract}.
In this paper we give a full proof of the abstract fixed-point theorem and describe its applications in detail.

\section{Preliminaries}\label{sec.preliminaries}

\subsection{Notations}

A \emph{language} $\mathcal{L}$ consists of constant, predicate and function symbols. \emph{Terms}
 over $\mathcal{L}$ are built from individual variables, constant and function symbols of $\mathcal{L}$.
\emph{First-order (FO) formulas} over $\mathcal{L}$ are built from predicate symbols, terms, the
 logical connectives $\neg, \vee, \wedge$ and the quantifiers $\exists, \forall$ over individual
 variables.
Moreover, we use the symbols $\rightarrow$ and $\leftrightarrow$, where $A \rightarrow B$ is defined
 to be an abbreviation for $\neg A \vee B$ and $A \leftrightarrow B$ to be an abbreviation for
 $A \rightarrow B \wedge B \rightarrow A$.
In a \emph{second-order (SO) formula}, in addition, there may occur predicate variables, which may be
 bound by universal or existential quantifiers.
Predicate variables stand for predicates of a certain arity.
To distinguish them from individual variables we denote them with upper-case letters.
Unless otherwise noted we always talk about logic \emph{with equality}, which means that
 we have a specified binary predicate symbol ``$=$'', which is interpreted as equality.

A subformula $\psi$ of a formula $\varphi$ occurs positively in $\varphi$ if $\psi$ occurs in the scope of an even number of negations in $\varphi$ and occurs negatively if it occurs in the scope of an odd number of negations. For example, in the formula $\varphi := A \rightarrow B$
the subformula $A$ occurs negatively and $B$ occurs positively in $\varphi$, as $\varphi$
is an abbreviation for $\neg A \vee B$. A predicate variable $X$ \emph{occurs only positively} in $\varphi$ if every occurrence of $X$ as a subformula in $\varphi$ occurs positively. Conversely $X$ \emph{occurs only negatively} if every occurrence of $X$ as a subformula in $\varphi$ occurs negatively.

For a formula $\varphi$, variables $x_1,...,x_n$ and terms $t_1,...,t_n$ we define
 $\varphi[x_1\backslash t_1,...,x_n \backslash t_n]$ to be the formula $\varphi$, where
 every occurrence of $x_j$ is replaced by $t_j$ for every $j \in \{1,...,n\}$ simultaneously.
If a free variable in $t_1,...,t_n$ also occurs as a bound variable in $\varphi$, then we consider
 the variant $\varphi'$, where every bound variable which occurs in $t_1,...,t_n$ is renamed.
Thus the substitution $\varphi[x_1\backslash t_1,...,x_n \backslash t_n]$ is always defined, with
 possible renaming of bound variables.
Similarly, for predicate variables $X_1,...,X_n$ and formulas
 $\alpha_1,...,\alpha_n$ we define $\varphi[X_1 \backslash \alpha_1,...,X_n \backslash \alpha_n]$ to
 be the formula $\varphi$, where every occurrence of $X_j$ is replaced by $\alpha_j$ for every
 $j \in \{1,...,n\}$.
We use the usual vector notation, i.e., we write $\mybar{X}$ for $X_1,...,X_n$ if $n$ is clear from
 the context or unimportant.
Consistently we write $\varphi[\mybar{X} \backslash \mybar{\alpha}]$ for $\varphi[X_1 \backslash \alpha_1,...,X_n \backslash \alpha_n]$.

An $\mathcal{L}$-\emph{structure} is a pair $\mathcal{M} = (M,I)$, where $M$ is a set and $I$ is an interpretation of $\mathcal{L}$, i.e. $I(P) \subseteq M^k$ for a $k$-ary predicate symbol $P \in \mathcal{L}$ and $I(f): M^k \rightarrow M$ for a $k$-ary function symbol $f \in \mathcal{L}$.
An \emph{environment} is an interpretation of free variables by elements of the
 structure.
For an environment $\theta$, a variable $x$ and $m \in M$ we define $\theta[x:=m]$ by
 $\theta[x:=m](x) = m$ and $\theta[x:=m](y) = \theta(y)$ for $x \neq y$. $\theta[X:=S]$
 is defined analogously for a $k$-ary predicate variable $X$ and $S \subseteq M^k$.
For a structure $\mathcal{M}$, an environment $\theta$ and a formula $\varphi$ we define $\mathcal{M},\theta \models \varphi$ as usual. In particular $\mathcal{M},\theta \models \exists X \varphi(X)$ for a $k$-ary predicate variable $X$ if there exists an $S \subseteq M^k$ such that $\mathcal{M},\theta[X:=S] \models \varphi(X)$. We also write $\calM,\theta \models \phi(S)$ as an abbreviation for $\mathcal{M},\theta[X:=S] \models \varphi(X)$, where $X$ is a fresh predicate variable.
Similarly, if $a \in M$, we write $\mathcal{M}, \theta \models \psi(a)$ as abbreviation
 for $\mathcal{M}, \theta[x:=a] \models \psi(x)$.
We define $\mathcal{M} \models \varphi$ if $\mathcal{M},\theta \models \varphi$ for all environments $\theta$. A formula is \emph{valid}, written as $\models \varphi$, if $\mathcal{M} \models \varphi$ for all structures $\mathcal{M}$. We write $\varphi \equiv \psi$, if $\varphi$ and $\psi$ are logically equivalent, i.e. if $\models \varphi \leftrightarrow \psi$.

We also consider atomic \emph{least fixed-point (LFP)} formulas of the form
\begin{align*}
	[\lfp_R ~\varphi(R,\mybar{x})](\mybar{t}), 
\end{align*}
where $\varphi(R,\mybar{x})$ is a first-order formula in $\calL \cup \{R\}$, such that $R$ occurs only positively in $\varphi$, R is $k$-ary and $\mybar{t}$ is a $k$-tuple of terms. The semantics of an atomic LFP formula will be defined in the next subsection.

In this paper we talk about the following formula classes:
\begin{enumerate}
	\item FO: classical first-order logic,
	\item {SO:} Second-order logic,
	\item {FO[LFP]}: \emph{Least fixed-point logic} is an extension of first-order logic, which in addition to the usual formation rules allows atomic LFP-formulas,
	\item SO[LFP]: \emph{Second-order least fixed-point logic} is a syntactic extension of second-order logic, which additionally allows atomic LFP formulas.
\end{enumerate}

\begin{example}
	Let $\calL = \{E,s\}$ be the language of graphs with edge relation $E$ and one specified vertex $s$. Then
	\begin{enumerate}
		\item $\varphi(R,x) := E(s,x) \vee \exists y (R(s,y) \wedge E(y,x))$ is a FO-formula, where $R$ is a binary predicate variable.
		\item $\forall R (\varphi(R,x) \impl E(s,x))$ is a SO-formula,
		\item $\psi(y) := [\lfp_R ~\varphi(R,x)](y)$ is a FO[LFP]-formula and
		\item $\exists P (\psi(y) \leftrightarrow \neg P(y))$ is a SO[LFP]-formula.
	\end{enumerate}
\end{example}

Monotonicity is a key property for defining least fixed-points.
Before we define the semantics of atomic least fixed-point formulas we state a simple
 monotonicity lemma that will be useful at several occasion later on.
\begin{lemma}\label{lem.ImplPositiv}
	Let $\calL$ be a language and $\calM$ be an $\calL$-structure. Let $R$ be a predicate variable and let $S$ and $S'$ be relations in $\calM$ with the same arity as $R$. Let $\varphi$ be a first-order formula in $\calL \cup \{R\}$, such that $R$ occurs only positively in $\varphi$. Then
	\begin{align*}
		\calM \models S(\mybar{x}) \rightarrow S'(\mybar{x}) \quad \Rightarrow \quad \calM \models \varphi[R \backslash S] \rightarrow \varphi[R \backslash S'].
	\end{align*}
	Conversely, if $R$ occurs only negatively in $\varphi$, then
	\begin{align*}
		\calM \models S(\mybar{x}) \rightarrow S'(\mybar{x}) \quad \Rightarrow \quad \calM \models \varphi[R \backslash S'] \rightarrow \varphi[R \backslash S].
	\end{align*}
\end{lemma}

\subsection{Fixed-point logics}

\subsubsection{Semantics of least fixed-point formulas}
For the semantics of least fixed-point logic we need some background on fixed points, which
 is best presented in the context of complete lattices.

\begin{definition}
	Let $(E,\leq)$ be a complete lattice.
	\begin{enumerate}
		\item  A function $f: E \impl E$ is called an \emph{operator} on $E$.
		\item $f$ is called \emph{monotone} if for all $x \leq y$ it holds that $f(x) \leq f(y)$.
	%	\item An operator is \emph{inflationary} if $X \subseteq F(X)$ for all $X \subseteq A$.
		\item An element $x$ is a \emph{fixed point} of $f$ if $f(x) = x$.
		\item An element $x$ is the \emph{least fixed point} of $f$, if $x$ is a fixed point of $f$ and for any fixed point $y$ of $f$ it holds that $x \leq y$. This is denoted as $x = \lfp(f)$.
	\end{enumerate}
\end{definition}
%
%Note that there are monotone operators, which are not inflationary and inflationary operators, which are not monotone.
%
We are particularly interested in the complete lattice $(M^k,\subseteq)$, where $M$ is
 the domain of a structure $\calM$.
Towards the definition of the semantics of least fixed-point atoms we start with defining
 the operator induced by a formula.
\begin{definition}\label{def.fixedPointOperator}
        Let $\calL$ be a language and $R$ a predicate variable of arity $k$. Let $\varphi(R,x_1,...,x_k)$ be a first-order formula in $\calL \cup \{R\}$. Note that there could also occur free variables in $\phi$, which are not stated explicitly.
        For an $\calL$-structure $\calM$ define the operator $F_{\varphi}$ on $M^k$ by
	\begin{align*}
		F_{\varphi}: \quad X \mapsto \{\mybar{a} \in M^k ~|~ \calM \models \varphi(X,\mybar{a}) \}.
	\end{align*}
\end{definition}
For the definition of $F_{\phi}$ in the case that $\phi$ has free variables, recall that we defined $\mathcal{M} \models \varphi$ if $\mathcal{M},\theta \models \varphi$ for all environments $\theta$.

\begin{lemma}\label{lfpmonotoneLemma}
	If $R$ occurs only positively in $\varphi$, then $F_{\varphi}$ is monotone.
\end{lemma}
\begin{proof}
	Let $A \subseteq B \subseteq M^k$. We have $\calM \models A(\mybar{x}) \rightarrow B(\mybar{x})$. Then Lemma \ref{lem.ImplPositiv} yields $\calM \models \varphi(A,\mybar{a}) \rightarrow \varphi(B,\mybar{a})$. Thus $F_{\varphi}(A) \subseteq F_{\varphi}(B)$.
\end{proof}

The following theorem is an important result in the study of complete lattices, a proof can be found in \cite{Arnold2001}. 

\begin{theorem}[Knaster-Tarski]\label{Knaster-Tarski}
Let $(E,\leq)$ be a complete lattice and $f$ be a monotone operator on $E$. Then $f$ has a least fixed point and 
	\begin{align*}
		\lfp(f) = \bigwedge \{x ~|~ f(x) \leq x\}.
	\end{align*} 
\end{theorem}

\begin{definition}[LFP]
	The semantics of an atomic least fixed-point formula is defined as follows:
	\begin{align*}
		\calM \models [\lfp_R ~\varphi(R,\mybar{x})] (\mybar{a}) :\Leftrightarrow~ \mybar{a} \in \lfp(F_{\varphi}).
	\end{align*}
\end{definition}

As $R$ occurs only positively, $F_{\varphi}$ is a monotone operator according to Lemma \ref{lfpmonotoneLemma}.
Using Theorem~\ref{Knaster-Tarski}, the least fixed point exists and thus the semantics
 of least fixed-point atoms is well-defined.

\begin{example}\label{lfpExample}
	Let $\calL = \{E\}$ be the language of graphs, where $E$ is a binary relation symbol
	representing the edge relation, and let $R$ be a binary predicate variable. Define
	\begin{align*}
		\varphi(R,u,v) \equiv E(u,v) \vee \exists w (R(u,w) \wedge E(w,v)).
	\end{align*}
	As $R$ occurs only positively in $\varphi$ we can define $[\lfp_R ~\varphi(R,u,v)] (x,y)$, which holds iff there is a path from $x$ to $y$.
\end{example}

\subsubsection{Simultaneous fixed point logics}

The previously introduced concepts may be generalised for product lattices. This leads to the notion of simultaneous fixed points. %We skip the proofs as they are similar, only more technical, to their counterparts in the previous subsection.

\begin{definition}
	Let $(E_1,\leq),...,(E_n,\leq)$ be complete lattices.
	\begin{enumerate}
		\item  A function  $f: E_1\times\cdots\times E_n \impl E_1\times\cdots\times E_n$ is called an \emph{$n$-ary operator} on $E_1\times\cdots\times E_n$.
		\item For two tuples of elements $\mybar{x} := (x_1,...,x_n)$ and $\mybar{y} :=(y_1,...,y_n)$ in $ E_1\times\cdots\times E_n$ we write $\mybar{x} \leq \mybar{y}$ if $x_i \leq y_i$ for all $i \in \{1,...,n\}$.
		\item $f$ is \emph{monotone} if for all $\mybar{x} \leq \mybar{y}$ it holds that $f(\mybar{x}) \leq f(\mybar{y})$.
		%\item $F$ is \emph{inflationary} if for all $\mybar{X}$ and $i \in \{1,...,n\}$ it holds that $X_i \subseteq F(\mybar{X})_i$.
		\item $\mybar{x}$ is a \emph{fixed point} of $f$ if $f(\mybar{x}) = \mybar{x}$.
		\item If $\mybar{x}$ is a fixed point of $f$ and for every fixed point $\mybar{y}$ we have $\mybar{x} \leq \mybar{y}$, then $\mybar{x}$ is called the \emph{least fixed point} of $f$, written $\mybar{x} = \lfp(f)$.
	\end{enumerate}
\end{definition}

Similarly to Theorem \ref{Knaster-Tarski} we have (cf. \cite{Arnold2001}):
\begin{theorem}
	Let $(E_1,\leq),...,(E_n,\leq)$ be complete lattices and $f$  be a monotone operator on $E_1\times\cdots\times E_n$. Then $f$ has a least fixed point.
\end{theorem}

\begin{definition}\label{def.SimFixedPointOperator}
	Let $\calL$ be a language and $R_1,...,R_n$ be predicate variables, with $R_i$ being of arity $k_i$. Consider an $n$-tuple of first-order formulas  $\Phi = (\varphi_i(R_1,\ldots,R_n,\mybar{x_i}))_{i=1}^n$ in  $\calL \cup \{R_1,...,R_n\}$, where $|\mybar{x_i}| = k_i$.

	For a structure $\calM$ define 
	\begin{align*}
		F_i^{\calM}:\quad  M^{k_1} \times \cdots \times M^{k_n} &\rightarrow M^{k_i},\\
		(X_1,...,X_n) &\mapsto \{\mybar{x} \in M^{k_i} ~|~ \calM \models \varphi_i(X_1,...,X_n,\mybar{x}) \}.
	\end{align*}
	Now define $F_{\Phi}^{\calM} = (F_1^{\calM},...,F_n^{\calM})$. If $\calM$ is clear from the context we write $F_{\Phi}$. Moreover we write $(F_{\Phi})_i$ for $F_i^{\calM}$.
\end{definition}

\begin{lemma}\label{lfpSimmonotone}
	If $R_1,...,R_n$ occur only positively in $\Phi$, then $F_{\Phi}$ is monotone.
\end{lemma}
\begin{proof}
	This proof is analogous to the proof of Lemma \ref{lfpmonotoneLemma}.
\end{proof}

\begin{definition}[LFP$^{SIM}$]
	Let $\calL$ be a language and $R_1,...,R_n$ be predicate variables with $R_i$ being of arity $k_i$. Let  $\Phi = (\varphi_i(R_1,..,.R_n,\mybar{x_i}))_{i=1}^n$ be an $n$-tuple of first-order formulas in  $\calL \cup \{R_1,...,R_n\}$, where $|\mybar{x_i}| = k_i$ and $R_1,..,R_n$ occur only positively in $\Phi$ .
	An atomic \emph{simultaneous least fixed-point(LFP$^{SIM}$)} formula is of the form
	\begin{align*}
		[\lfp_{R_i} ~\Phi](\mybar{t}), 
	\end{align*}
 where $\mybar{t}$ is a $k_i$-tuple of terms in $\calL$.
	The semantics is defined as follows:
	\begin{align*}
		\calM \models [\lfp_{R_i} ~\Phi](\mybar{a}) :~\Leftrightarrow~ \mybar{a} \in \lfp(F_{\Phi})_i.
	\end{align*}
\end{definition}

\begin{example}
	Again consider the language of graphs $\calL = \{E\}$, where $E$ is a binary relation symbol. Let $R$ and $S$ be two binary predicate variables. Define $\Phi$ as
	\begin{align*}
		\varphi_1(R,S,u,v) &\equiv E(u,v) \vee \exists w (S(u,w) \wedge E(w,v)),\\
		\varphi_2(R,S,u,v) &\equiv \exists w (R(u,w) \wedge E(w,v)).
	\end{align*}
	As $R$ and $S$ occur only positively in $\Phi$, we can define the atomic LFP$^{SIM}$ formulas $[\lfp_R ~\Phi](x,y)$ and $[\lfp_S ~\Phi](x,y)$. The former holds iff there is a path of odd length and the latter iff there is an path of even length from $x$ to $y$.
\end{example}

Define FO[LFP$^{SIM}$] to be an extension of first-order logic, which also allows atomic LFP$^{SIM}$ formulas. Analogously SO[LFP$^{SIM}$] is an extension of second-order logic, which additionally allows atomic LFP$^{SIM}$ formulas.

Note that  FO[LFP$^{SIM}$], while being more convenient, is not more expressive than FO[LFP]
 in the following sense:
For every formula $\varphi$ in FO[LFP$^{SIM}]$ there exists a formula $\psi$ in FO[LFP],
 such that $\varphi \equiv \psi.$
A proof of this result can be found in \cite{Libkin2004}.
%
% \begin{theorem}\label{leastFixedPointLogicSimultaneousEquivalent}
% 	For every formula $\varphi$ in FO[LFP$^{SIM}]$ there exists a formula $\psi$ in FO[LFP], such that $\varphi \equiv \psi.$
% \end{theorem}
%
Note that the converse of this result is also true, as FO[LFP$^{SIM}]$ is a generalisation
 of FO[LFP].
Thus we can use LFP and LFP$^{SIM}$ interchangeably.
It is also well-known that the LFP-operator can be expressed in second-order logic so that we have:  For every FO[LFP] formula $\varphi$
 there exists a SO formula $\psi$ such that $\varphi  \equiv \psi$.

%
% \begin{lemma}
% 	Let $\calL$ be a language and $\varphi$ be a FO[LFP] formula. Then there exists a SO formula $\psi$ such that $\varphi  \equiv \psi$.
% \end{lemma}
% \begin{proof}
% 	By induction on the formula $\varphi$. The only case we have to check is an atomic formula of the form 
% 	\begin{align*}
% 		\varphi(\mybar{y}) \equiv [\lfp_R ~\varphi(R,\mybar{x})](\mybar{y}).
% 	\end{align*}
% 	Define
% 	\begin{align*}
% 		\psi(\mybar{y}) \equiv \exists X [& \forall \mybar{x} (\varphi(X,\mybar{x}) \impl X(\mybar{x})) \\
% 		& \wedge \forall Y [\forall \mybar{x} (\varphi(Y,\mybar{x}) \impl Y(\mybar{x})) \impl \forall \mybar{x} ( X(\mybar{x}) \impl Y(\mybar{x}))] \\
% 		&\wedge X(\mybar{y})], 
% 	\end{align*}
% 	then $\varphi \equiv \psi$.
% \end{proof}
% \hrule

\subsubsection{Fixed-point approximation}\label{sec.prelim.approximation}
We are particularly interested in least fixed-point formulas defined from existential first-order formulas.
We call a first-order formula $\varphi$ existential, if there exists a formula $\tilde{\varphi}$ in
prenex normal form without universal quantifiers, that is logically equivalent to $\varphi$.
For many applications, in particular in software verification, it suffices to consider least fixed-points
of existential formulas.
This point has been made prominently in~\cite{Blass87Existential}.
Those least fixed-point formulas may be approximated by first-order logic formulas. In order to achieve this we first define infinite conjunctions and disjunctions:

\begin{definition}\label{def.infiniteFormulas}
	Let $\calL$ be a language and $\calM$ be an $\calL$-structure. Let $\Psi$ be a set of first-order formulas in $\calL$. Define
	\begin{align*}
		\calM \models \bigvee_{\varphi \in \Psi} \varphi ~&:\Leftrightarrow~ \exists \varphi \in \Psi :~ \calM \models \varphi, \\
		\calM \models \bigwedge_{\varphi \in \Psi} \varphi ~&:\Leftrightarrow~ \forall \varphi \in \Psi :~ \calM \models \varphi,
	\end{align*}
\end{definition}

\begin{lemma}\label{lem.infiniteFormulas}
	Let $\calL$ be a language and let $\Psi_1,\Psi_2$ be sets of first-order formulas in $\calL$. Then
	\begin{enumerate}
		\item $ \neg \bigvee_{\varphi \in \Psi_1} \varphi \equiv \bigwedge_{\varphi \in \Psi_1} \neg \varphi$,
		\item  $\bigwedge_{\varphi_1 \in \Psi_1} \varphi_1 \vee \bigwedge_{\varphi_2 \in \Psi_2} \varphi_2 \equiv \bigwedge_{(\varphi_1,\varphi_2) \in \Psi_1\times\Psi_2} \varphi_1 \vee \varphi_2$,
		\item  $\bigwedge_{\varphi_1 \in \Psi_1} \varphi_1 \wedge \bigwedge_{\varphi_2 \in \Psi_2} \varphi_2 \equiv \bigwedge_{\varphi \in \Psi_1\cup\Psi_2} \varphi$,
		\item $\forall y \bigwedge_{\varphi \in \Psi_1} \varphi(y) \equiv \bigwedge_{\varphi \in \Psi_1} \forall y ~\varphi(y)$.
	\end{enumerate}
\end{lemma}

For existential first-order formulas, we can define a sequence of sets, that converges to the least fixed point. These sets may be described by first-order formulas.

\begin{definition}
		Let $(E_1,\leq),...,(E_n,\leq)$ be complete lattices and $f$ be an $n$-ary operator on $E_1 \times \cdots \times E_n$.  By induction define the following elements in $E_1\times\cdots\times E_n $:
	\begin{align}\label{ifpSequenceDefinition}
		S_f^0 &= (\bot,...,\bot),\nonumber\\
		S_f^{n+1} &= f(S_f^{n}),\\
		S_f^{\omega} &= \bigcup_{n \in \omega} S_f^{n}.\nonumber
	\end{align}
If it is obvious which operator we talk of, we sometimes omit the subscript and write $S^{n}$ and $S^\omega$.
\end{definition}

We are interested in the operator $F_{\Phi}$ of Definition \ref{def.SimFixedPointOperator}. If $F_{\Phi}$ is defined from existential formulas, then $F_{\Phi}$ is continuous. In this case Kleene's fixed point theorem states that $S^\omega$ coincides with the least fixed point of $F_{\Phi}$. The lemma also follows from \cite[Theorem 9]{Blass87Existential}:
\begin{lemma}\label{lem.apprClosureOrdinalOmega}
		Let $\calL$ be a language and $R_1,...,R_n$ be predicate variables. Let $\Phi =  (\varphi_i(R_1,...,R_n,\mybar{x_i}))_{i=1}^n$ be an $n$-tuple of existential first-order formulas, such that $R_1,...,R_n$ occur only positively in $\Phi$. Let $\calM$ be an $\calL$-structure and $f = F_{\Phi}^{\calM}$. Then $S_f^\omega = \lfp(f)$.
\end{lemma}
% \begin{proof}
% 	For notational simplicity, we prove the lemma for one predicate variable $R$, the general proof is analogous. Let $\calM$ be a model and let $\mybar{a} \in F_{\varphi}(S_{F_{\varphi}}^{\omega})$. This means $\calM,[R:= S_{F_{\varphi}}^{\omega}] \models \varphi(R,\mybar{a})$. As $\varphi$ is existential it can be written as $\varphi(R,\mybar{a}) \equiv \exists \mybar{y} ~\tilde{\varphi}(R,\mybar{a},\mybar{y})$, where $\tilde{\varphi}$ is quantifier-free. Choose $\mybar{y_0}$ such that 
% 	\[
% 	\calM,[R:= S_{F_{\varphi}}^{\omega},\mybar{y} := \mybar{y_0}] \models \tilde{\varphi}(R,\mybar{a},\mybar{y}).
% 	\]
% 	In $\tilde{\varphi}$ there are only finitely many occurrences of $R(\mybar{t}(\mybar{a},\mybar{y}))$, where $\mybar{t}$ is a tuple of terms in $\calL$. As $S_{F_{\varphi}}^{\omega} = \bigcup_{l \in \omega} S_{F_{\varphi}}^{l}$, there is an $l_0 \in \omega$ such that these occurrences of $R$ interpreted as $S_{F_{\varphi}}^{\omega}$ and $S_{F_{\varphi}}^{l_0}$ yield the same truth value. Hence
% 		\[
% 	\calM,[R:= S_{F_{\varphi}}^{l_0},\mybar{y} := \mybar{y_0}] \models \tilde{\varphi}(R,\mybar{a},\mybar{y})
% 	\]
% 	and therefore $\mybar{a} \in F_{\varphi}(S_{F_{\varphi}}^{l_0}) \subseteq F_{\varphi}^{\omega}$. The other inclusion $S_{F_{\varphi}}^{\omega} \subseteq F_{\varphi}(S_{F_{\varphi}}^{\omega})$ holds, as $F_{\varphi}$ is monotone.
% \end{proof}

The next step is to describe the sets $S^{n}$ and $S^\omega$ with first-order formulas.
\begin{definition}\label{def.apprSigma}
	Let $\calL$ be a language and $R_1,...,R_n$ be predicate variables. Let $\Phi =  (\varphi_i(R_1,...,R_n,\mybar{x_i}))_{i=1}^n$ be an $n$-tuple of first-order formulas. For $j \in \{1,...,n\}$ define
	\begin{align*}
		\sigma_{j,\Phi}^0(\mybar{x_j}) &\equiv \bot\\
		\sigma_{j,\Phi}^{l+1}(\mybar{x_j}) &\equiv \varphi_j(\sigma_{1,\Phi}^l,...,\sigma_{n,\Phi}^l,\mybar{x_j})\\
		\sigma_{j,\Phi}^{\omega}(\mybar{x_j}) &\equiv \bigvee_{l\in \omega} \sigma_{j,\Phi}^l(\mybar{x_j}).
	\end{align*}
\end{definition}

\begin{lemma}\label{lem.apprEquivSetsFormulas}
	Let $\calL$ be a language and let $R_1,...,R_n$ be predicate variables. Let $\Phi =  (\varphi_i(R_1,...,R_n,\mybar{x_i}))_{i=1}^n$ be an $n$-tuple of first-order formulas, such that $R_1,...,R_n$ occur only positively in $\Phi$. Let $\calM$ be an $\calL$-structure. For every $j \in \{1,...,n\}$ and $\mybar{a} \in M^{k_j}$, where $k_j$ equals the arity of $R_j$, we have
	\begin{align*}
		\calM \models \sigma_{j,\Phi}^l(\mybar{a}) \quad \Leftrightarrow \quad  \mybar{a} \in (S_{F_{\Phi}}^l)_j
	\end{align*}
	for every $l \in \omega \cup \{\omega\}$.
\end{lemma}
\begin{proof}
	This follows by induction on $l$ from the definition of $\sigma_{j,\Phi}^l$ and $S_{F_{\Phi}}^{l}$.
\end{proof}

\begin{theorem}\label{thm.apprLfp}
		Let $\calL$ be a language and let $R_1,...,R_n$ be predicate variables. Let $\Phi =  (\varphi_i(R_1,...,R_n,\mybar{x_i}))_{i=1}^n$ be an $n$-tuple of existential first-order formulas, such that $R_1,...,R_n$ occur only positively in $\Phi$. Then for all $j \in \{1,...,n\}$
		\[
		[\lfp_{R_j} ~\Phi] \equiv \sigma_{j,\Phi}^{\omega}.
		\]
\end{theorem}
\begin{proof}
	Let $\calM$ be an $\calL$-structure and $\mybar{a} \in M^{k_j}$. Lemma \ref{lem.apprClosureOrdinalOmega} states that $\calM \models [\lfp_{R_j} ~\Phi](\mybar{a})$ iff $\mybar{a} \in (S_{F_{\Phi}}^{\omega})_j$ for $j \in \{1,...,n\}$. Now Lemma \ref{lem.apprEquivSetsFormulas} concludes $\calM \models [\lfp_{R_j} ~\Phi](\mybar{a}) \leftrightarrow \sigma_{j,\Phi}^{\omega}(\mybar{a})$ for $j \in \{1,...,n\}$.
\end{proof}

\section{Formula equations}\label{sec.form_eq}

\subsection{Introduction to formula equations}

Equations are pervasive in mathematics.
Solving a given equation means to find objects for the unknowns such that, after substitution,
 the two sides of the  equation are equal.
What kinds of objects are permissible as solutions and what ``equal'' means depends on the
 type of equations under consideration.
Occasionally equations can be brought into normal forms in which one of the two sides becomes
 trivial.
This is, for example, the case for systems of linear equations.

In this paper, we work with formula equations in first-order logic.
As a formula equation we want to consider
\begin{equation}\label{eq.feqtwosided}
\text{an $\calL$-formula}\ \psi_1 \liff \psi_2\ \text{containing predicate variables $\mybar{X} = X_1,\ldots,X_n$.}
\end{equation}
The kind of objects we want to allow as solutions are first-order formulas.
The notion of equality we want to satisfy after inserting a solution for the
 predicate variables is logical equivalence.
Consequently, a solution of~(\ref{eq.feqtwosided}) is a first-order substitution
 $\unsubst{\mybar{X}}{\mybar{\chi}}$ such that\ $\models \psi_1\unsubst{\mybar{X}}{\mybar{\chi}} \liff \psi_2 \unsubst{\mybar{X}}{\mybar{\chi}}$.
Formula equations can easily be brought into a form in which one of the two sides is trivial:
we can simplify formula equations by instead considering
\begin{equation}\label{eq.feqonesided}
\text{an $\calL$-formula}\ \psi\ \text{containing predicate variables $\mybar{X} = X_1,\ldots,X_n$.}
\end{equation}
Then, as a solution, we ask for a first-order substitution $\unsubst{\mybar{X}}{\mybar{\chi}}$ such that\ 
 $\models \psi\unsubst{\mybar{X}}{\mybar{\chi}}$.
Note that every instance of~(\ref{eq.feqtwosided}) is an instance of~(\ref{eq.feqonesided}) by letting $\psi$ be $\psi_1 \liff \psi_2$
 and every instance of~(\ref{eq.feqonesided}) is an instance of~(\ref{eq.feqtwosided}) by letting $\psi_1$ be $\psi$ and 
 $\psi_2$ be $\top$.
Moreover, it will be notationally useful to explicitly indicate the predicate variables by 
 existential quantifiers.
Consequently we define:
\begin{definition}
A {\em formula equation} is a closed $\calL$-formula $\exists \mybar{X}\, \psi$ where
 $\psi$ contains only first-order quantifiers.
A {\em solution} of $\exists \mybar{X}\, \psi$ is a first-order substitution
 $\sop\sel{\mybar{X}}{\mybar{\chi}}\scl$ such that\  $\models \psi\unsubst{\mybar{X}}{\mybar{\chi}}$. We call a formula equation \emph{solvable} if there exists a solution. 
\end{definition}
Note that a formula equation is simply a closed existential second-order formula.
We prefer to work with the terminology ``formula equation'' because 1.~it is shorter
 and 2.~it emphasises the aspect of finding a solution which is central for this entire work.
\begin{example}\label{ex.runex_intro}
We will now start a more precise treatment of the formula equation
\[
 \exists X\ (X(2) \land \forall u\, (X(u) \impl X(u+u)) \land \neg X(3))
\]
mentioned in the introduction: we work in the language $\calL_A = \{0,s,+,<\}$ of linear arithmetic.
$2$ is an abbreviation for the term $s(s(0))$, $3$ is an abbreviation of $s(s(s(0)))$, and so on.
We need a simple background theory to establish basic facts, such as, for example,
 $\neg \exists x\, 3 = x + x$.
To that aim let $T$ be the theory of open induction for $\calL_A = \{0,s,+,<\}$.
Then we define the formula equation
\[
\exists X\, \big( T \impl (X(2) \land \forall u\, (X(u) \impl X(u+u)) \land \neg X(3)) \big)
\]
Let $\chi_1(v) \equiv \exists w\, (v = w + w)$ and $\chi_2(v) \equiv  v \neq 3$. Then $\unsubst{X}{\chi_1}$ and $\unsubst{X}{\chi_2}$ are solutions.
\end{example}
The problem of computing a solution of a formula equation given as input will be denoted as FEQ in the sequel.
A formula equation $\exists \mybar{X}\, \psi$ is called {\em valid} if it is a valid second-order formula
 and {\em satisfiable} if it is a satisfiable second-order formula.
Every solvable formula equation is valid and every valid formula equation is satisfiable but neither of
 the converse implications are true as the following example shows.
\begin{example}\label{ex.feq}
If $\psi$ is a first-order formula which is satisfiable but not valid and does not contain $X$
 then, trivially, $\exists X\, \psi$ is a formula equation which is satisfiable but not valid.
 
Towards an example for a valid but unsolvable formula equation we work in the first-order language
 $\calL = \{ 0, s \}$, where $0$ is a constant symbol and $s$ is a unary function symbol.
Let $A_1$ be $\forall x\, s(x)\neq 0$, let $A_2$ be $\forall x\forall y\, (s(x)=s(y)\impl x=y)$
 and consider the formula
\[
A_1 \land A_2 \impl \exists X \exists Y \forall u\, \Big( X(0) \land Y(s(0)) \land (X(u) \impl Y(s(u)))
 \land (Y(u) \impl X(s(u))) \land \neg (X(u) \land Y(u)) \Big)
\]
which, up to some simple logical equivalence transformations, is a formula equation $\Phi$.
Now $\Phi$ is valid since, in a model $\calM$ of $A_1 \land A_2$, interpreting $X$ by
 $\{ s^{2n}(0)^\calM \mid n \in \Nat \}$ and $Y$ by $\{ s^{2n+1}(0)^\calM \mid n \in \Nat \}$ makes the remaining
 formula true.

For unsolvability suppose that $\Phi$ has a solution $\sop \sel{X}{\chi(u)}, \sel{Y}{\psi(u)} \scl$.
Then, since the standard model $\Nat$ in the language $\calL$ satisfies $A_1\land A_2$, we would have
\[
 \Nat \models \chi(0)\land \psi(s(0)) \land \forall u\, (\chi(u) \impl \psi(s(u))) \land \forall u\, (\psi(u) \impl \chi(s(u))) \land \forall u\, \neg (\chi(u) \land \psi(u)),
\]
in particular, $\chi$ would be a definition of the even numbers and $\psi$ a definition of the odd numbers.
However, the theory of $\Nat$ in $\calL$ admits quantifier elimination~\cite[Theorem 31G]{Enderton01Mathematical},
 which has the consequence that the $L$-definable sets in $\Nat$ are the finite and co-finite
 subsets of $\Nat$~\cite[Section 3.1, Exercise 4]{Enderton01Mathematical}.
Thus we obtain a contradiction to $\chi$ being
 a definition of the even numbers (which is neither finite nor co-finite).
\end{example}

% FEQ/SOQE/WSOQE
Solving formula equations (FEQ) is closely related to the problem of second-order quantifier elimination (SOQE):
given a formula $\exists \mybar{X}\, \psi$ where $\psi$ contains only first-order quantifiers, find
 a first-order formula $\phi$ such that\ $\models \exists \mybar{X}\, \psi \liff \phi$.
The relationship between FEQ and SOQE is often based on a third problem:
 second-order quantifier elimination by a witness (WSOQE): given a formula $\exists \mybar{X}\, \psi$ where $\psi$ contains only first-order
 quantifiers, find a first-order substitution $\unsubst{\mybar{X}}{\mybar{\chi}}$
 such that\ $\models \exists \mybar{X}\, \psi \liff \psi\unsubst{\mybar{X}}{\mybar{\chi}}$.
Such formulas $\mybar{\chi}$ are called ELIM-witnesses in~\cite{Wernhard17Boolean}.
Such a tuple $\mybar{\chi}$ of formulas is a canonical witness for $\exists \mybar{X}\, \psi$
 in the sense that, whenever there is a tuple of sets $\mybar{X}$ satisfying $\psi$ then
 $\mybar{\chi}$ does.
The existence of (computable) canonical witnesses has two useful corollaries:
 1.~SOQE can be solved in a straightforward way and 2.~FEQ boils down to a validity check.
This has been exploited, for example, in the proof of decidability of the QFBUP
 problem~\cite{Eberhard17Boolean}.
Our fixed point theorem can be seen as precisely such a result on the existence of a canonical
 witness in FO[LFP].

The complex of these three problems, FEQ, SOQE, and WSOQE, has a long history in logic and
 a wealth of applications in computer science,
see the textbook~\cite{Gabbay08Second} on second-order quantifier elimination.
A number of algorithms for second-order quantifier elimination have been developed, for example:
The SCAN algorithm introduced in~\cite{Gabbay92Quantifier} (tries to) compute(s) a first-order formula equivalent to $\exists X\, \psi$
 for a conjunctive normal form $\psi$ by forming the closure of $\psi$ under constraint resolution which
 only resolves on $X$-literals.
The DLS algorithm has been introduced in~\cite{Doherty97Computing} and consists essentially of
 formula rewriting steps tailored to allow application of Ackermann's lemma (which instantiates a predicate
 variable provided some conditions on the polarity of its occurrences are met). This approach has been
generalised by Nonnengart and Sza{\l}as in \cite{Nonnengart98Fixpoint} by allowing FO[LFP], which resulted in the DLS* algorithm.

\subsection{Horn formula equations}

Towards the definition of Horn formula equations we first introduce the notion of constrained clause.
 \begin{definition}
	Let $\calL$ be a first-order language. A \emph{constrained clause} is a formula $C$ of the form
	\[
	\gamma \vee \bigvee_{i =1}^m \neg X_i(\mybar{t_i}) \vee \bigvee_{j=1}^n Y_j(\mybar{s_j}),
	\]
	where $\gamma$ is a first-order formula in $\calL$, $X_i,Y_j$ are predicate variables and $\mybar{t_i},\mybar{s_j}$ are tuples of terms in $\calL$ for $i \in \{1,...m\}, j \in \{1,...,n\}$. $C$ is called
	\begin{enumerate}
		\item \emph{Horn}, if $n \leq 1$,
		\item \emph{dual-Horn}, if $m \leq 1$ and
		\item \emph{linear-Horn}, if $m,n \leq 1$.
	\end{enumerate}
\end{definition}

Note that a constrained clause is allowed to (and typically does) contain free individual variables which,
 as usual in clause logic, are treated as universally quantified.
A finite set $S$ of constrained clauses is considered as the conjunction of these clauses and is thus
 logically equivalent to a formula of the form $\forall^* \Land_{C\in S} C$ where $\forall^*$ denotes
 the universal closure w.r.t.\ individual variables.
 
 \begin{definition}
 	A formula equation $\exists \mybar{X} \forall^* \Land_{C\in S} C$ is called a
 	\begin{enumerate}
 		\item \emph{Horn formula equation}, if $S$ is a set of constrained Horn clauses,
 		\item \emph{dual-Horn formula equation}, if $S$ is a set of constrained dual-Horn clauses and
 		\item \emph{linear-Horn formula equation}, if $S$ is a set of constrained linear-Horn clauses.
 	\end{enumerate}

 \end{definition}

Thus Horn formula equations correspond to existential second order Horn logic, which also
 plays a significant role in finite model theory, see~\cite{Graedel91Expressive}.

\begin{example}\label{ex.runex}
The formula equation
\[
\exists X\, \big( T \impl (X(2) \land \forall u\, (X(u) \impl X(u+u)) \land \neg X(3) \big)
\]
from Example~\ref{ex.runex_intro} is logically equivalent to the Horn formula equation
\[
\exists X\, \psi \equiv \exists X\, \big( ( T \impl X(2) ) \land \forall u\, (T \land X(u) \impl X(u+u)) \land (T\impl \neg X(3)) \big).
\]
We will use this Horn formula equation as running example throughout most of this paper.
\end{example}

There are different notions of solvability for constrained Horn clauses in the literature:
{\em satisfiability} of~\cite{Gurfinkel19Science} is satisfiability of a Horn formula equation,
{\em semantic solvability} of~\cite{Ruemmer13Classifying} is validity of a Horn formula equation, and
{\em syntactic solvability} of~\cite{Ruemmer13Classifying} is solvability of a Horn formula equation.
In this paper we will primarily be interested in this last notion: solvability of a Horn formula equation.
Note that the formula equation considered in Example~\ref{ex.feq} is Horn so that, also when restricted
 to the class of Horn formula equations, solvability implies validity and validity implies satisfiability but
 both inverse implications are not true.

Let $\exists \mybar{X}\, \psi$ be a Horn formula equation.
We distinguish three different types of clauses in $\psi$:

\begin{align*}
	\begin{array}{crl}
		\text{(B)} &\gamma &\impl Y(\mybar{s}),   \\
		\text{(I)} &\gamma \land X_{1}(\mybar{t_1}) \land \cdots \land X_{m}(\mybar{t_m}) &\impl Y(\mybar{s}),  \\
		\text{(E)} &\gamma \land X_{1}(\mybar{t_1}) \land \cdots \land X_{m}(\mybar{t_m}) &\impl \bot, 
	\end{array}
\end{align*}
where the constraint $\gamma$ is a formula in $\calL$ not containing predicate variables,
$m \geq 1$, $\mybar{t_1},..,\mybar{t_m},\mybar{s}$ are
tuples of first-order terms in $\calL$ of appropriate arity and $Y,X_1,...,X_m$ are predicate variables. 
Note that free variables $\mybar{y}$ may occur in the formulas $\gamma$ and the terms
$\mybar{s}, \mybar{t_1},\ldots,\mybar{t_m}$.
We call the first {\em base clauses}, the second {\em induction clauses}, and the third {\em end clauses}.

For simplicity let us first consider a Horn formula equation $\exists X\, \psi$ containing
 only one predicate variable $X$.
For finding a solution $R$ for $X$ in a given model $\calM$, it is natural to proceed
 as follows:
for any base clause $\gamma \impl X(\mybar{s})$, if $\calM \models \gamma$, then add
 $\mybar{s}$ to $R$.
Inductively, for any induction clause $\gamma \land X(\mybar{t_1}) \land \cdots \land
 X(\mybar{t_m}) \impl X(\mybar{s})$, if $\calM \models \gamma$ and $\mybar{t_1},...,\mybar{t_m}$ are already in $R$, then add $\mybar{s}$ to $R$.
In the rest of this section, we lay the groundwork for pushing this procedure
 on the object level.
We translate the clauses of the form (B) and (I) to formulas, where all
 predicate variables occur only positively:
Let $C$ be an induction clause of the form $\gamma \land X_{1}(\mybar{t_1}) \land \cdots
 \land X_{m}(\mybar{t_m}) \impl Y(\mybar{s})$.
We define the translated clause $T_C(\mybar{x})$ to be $\gamma \land X_{1}(\mybar{t_1}) \land \cdots \land X_{m}(\mybar{t_m}) \land \mybar{x} = \mybar{s}$.
Similarly for a base clause $C$ of the form $\gamma \impl Y(\mybar{s})$ the translated clause $T_C(\mybar{x})$ is $\gamma \land \mybar{x} = \mybar{s}$.
\begin{example}
Continuing the running example, Example~\ref{ex.runex}, let $C_1$ be the base clause $T \impl X(2)$
 and $C_2$ be the induction clause $T\land X(u) \impl X(u+u)$.
Then $T_{C_1}(x) \equiv T\land x = 2$ and $T_{C_2}(x) \equiv T\land X(u) \land x = u + u$.
\end{example}
Next we gather all clauses having the same positive predicate variable. For a base or induction clause $C$ let $P_C$ be the unique predicate variable that occurs positively. Let $\exists X_1 \cdots \exists X_n\, \psi$ be a Horn formula equation with predicate variables $X_1,...,X_n$. We define $B_j$ and $I_j$ to be the set of base and induction clauses $C$ in $\psi$, such that $P_C = X_j$, for $j = 1,...,n$.
\begin{definition}\label{def.HornformulaSplit}
	Let $\exists X_1 \cdots \exists X_n\, \psi$ be a Horn formula equation.
	Define $\Phi_{\psi}$ to be the $n$-tuple of first-order formulas $(\varphi_1,...,\varphi_n)$ defined as
	\[
	\varphi_j(X_1,\ldots,X_n,\mybar{x_j})
	:=
	\exists \mybar{y} 
	\Lor_{C \in B_j \cup I_j} T_C(\mybar{x_j}), 
	\]
	where $\mybar{y}$ are the free variables occurring in the clauses in $B_j \union I_j$ and $\mybar{x_j}$ is a
	tuple of variables such that $|\mybar{x_j}|$ equals the arity of $X_j$. If the context is clear we sometimes omit the subscript and write $\Phi$ instead of $\Phi_\psi$ to achieve clearer notation.
\end{definition}

\begin{example}\label{ex.runningPhi}
Continuing the running example, Example~\ref{ex.runex}, $\Phi_\psi$ is a $1$-tuple consisting of
\[
\phi(X,x) \equiv \exists u\, (( T \land x=2) \lor (T\land X(u) \land x = u+u)).
\]
which is logically equivalent to
\[
T \land \big( (x=2) \lor \exists u\, ( X(u) \land x = u+u) ).
\]
\end{example}

In $\Phi_\psi$ all predicate variables occur only positively.
Hence it defines a monotone operator which has a least fixed point.
From the point of view of (constraint) logic programming, the above
 tuple of formulas is a first-order definition
of the operator $T_P$ induced by $\exists \mybar{X}\, \psi$ when considered as a constraint logic program $P$, see, e.g.,~\cite{Jaffar94Constraint}.

\begin{lemma}\label{lem.Hornformula.split}
	Let  $\exists X_1 \cdots \exists X_n\, \psi$ be a Horn formula equation. Let $E$ be the set of end clauses in $\psi$ and $\Phi_{\psi}\ =(\varphi_1,...,\varphi_n)$. Then
	\[
	\psi \equiv \bigwedge_{j=1}^n \forall \mybar{x_j}\, (\varphi_j(X_1,...,X_n,\mybar{x_j}) \impl X_j(\mybar{x_j})) \wedge \forall \mybar{y} \bigwedge_{C \in E} C.
	\]
\end{lemma}
\begin{proof}
It suffices to show that, for all $j \in \{ 1,\ldots,n \}$, for all models $\calM$, and
 for all environments $\theta$ interpreting $X_1,\ldots,X_n$:
\begin{align*}
 \calM, \theta & \models \forall \mybar{x_j}\, (\varphi_j(X_1,...,X_n,\mybar{x_j}) \impl X_j(\mybar{x_j}))\ \text{iff} \\
 \calM, \theta & \models C\ \text{for all induction clauses $C$ in $\psi$ with head symbol $X_j$.}
\end{align*}
%
%
% 	We have the same end clauses on the right and the left side of the equivalence, thus we only have to consider the base and induction clauses of $\psi$. Let $\calM$ be a model and $\theta$ an environment of $X_1,...,X_n$. Let $C$ be an induction clause $\gamma \land Z_{1}(\mybar{t_1}) \land \cdots \land Z_{m}(\mybar{t_m}) \impl X_j(\mybar{s})$ in $\psi$, where $Z_1,...,Z_m$ are in $\{X_1,...,X_n\}$.

First assume that $\calM,\theta \not\models C$ for some induction clause $C$ with head
 symbol $X_j$, i.e.,  $\calM, \theta \models \exists \mybar{y} (\gamma \land \Land_{k = 1}^{m}  Z_{k}(\mybar{t_k}) \wedge \neg X_{j}(\mybar{s}))$.
	Then $T_C(\mybar{x}) \equiv \gamma \land\Land_{k = 1}^{m}  Z_{k}(\mybar{t_k}) \land \mybar{x} = \mybar{s}$ and therefore $\calM, \theta \models \exists \mybar{y}\, T_C(\mybar{s})$.
	Thus it holds $\calM, \theta \models (\varphi_j(X_1,...,X_n,\mybar{s}) \wedge \neg X_j(\mybar{s}))$ and hence $\calM, \theta \not\models \forall \mybar{x_j} (\varphi_j(X_1,...,X_n,\mybar{x_j}) \impl X_j(\mybar{x_j}))$. The analogous argument holds for base clauses.
	
	For the other direction assume that $\calM,\theta \not\models \bigwedge_{j=1}^n \forall \mybar{x_j} (\varphi_j(X_1,...,X_n,\mybar{x_j}) \impl X_j(\mybar{x_j}))$, i.e. $\calM,\theta \models \exists \mybar{x_j} (\varphi_j(X_1,...,X_n,\mybar{x_j}) \land \neg X_j(\mybar{x_j}))$ for some $j \in \{1,...,n\}$. Thus there is a clause $C \in B_j \cup I_j$ such that $\calM,\theta \models \exists \mybar{x_j}(\exists \mybar{y}\, T_C(\mybar{x_j}) \land \neg X_j(\mybar{x_j}))$. Let $C$ be an induction clause, the same argument also holds for base clauses. Then $T_C(\mybar{x}) \equiv \gamma \land\Land_{k = 1}^{m}  Z_{k}(\mybar{t_k}) \land \mybar{x} = \mybar{s}$. It follows that $\mybar{x_j} = \mybar{s}$ and $\calM,\theta \models \exists \mybar{y} (\gamma \land \Land_{k = 1}^{m}  Z_{k}(\mybar{t_k}) \land \neg X_j(\mybar{s}))$. Hence $\calM,\theta \not\models C$.
\end{proof}
The idea of the above transformation is to split the positive and negative occurrences of $X_1,...,X_n$ in $\psi$, note that $X_1,...,X_n$  occur only positively in $\varphi_1,...,\varphi_n$ and only negatively in the end clauses $C \in E$. This will be taken advantage of later.

\section{The abstract fixed-point theorem}\label{sec.afpthm}

The aim of this section is to prove the main result of this paper, the
 following abstract fixed-point theorem. 
 It states that Horn formula equations have a least solution.
\begin{theorem}[Abstract Horn fixed-point theorem]\label{thm.abstractFP}
	Let $\exists \mybar{X} \psi$ be a Horn formula equation and $\mu_j := [\lfp_{X_j} ~\Phi_{\psi}]$ for $j \in \{1,...,n\}$, then:
	\begin{enumerate}
		\item $\modelsa \exists \mybar{X}  ~\psi \leftrightarrow \psi[\mybar{X}\backslash \mybar{\mu}]$ and
		\item if $(\calM,G) \modelsa \psi[\mybar{X}\backslash \mybar{R}]$ for some model abstraction $(\calM,G)$ and abstract relations $R_1,...,R_n$, then $(\calM,G) \modelsa \bigwedge_{j=1}^n (\mu_j \rightarrow R_j)$.
	\end{enumerate}
\end{theorem}

Here $\Phi_\psi$ is the tuple of first-order formulas induced by
 $\exists \mybar{X}\, \psi$ as in Definition~\ref{def.HornformulaSplit}.
The theorem is formulated for abstract semantics $\modelsa$
 which is a generalization of standard semantics $\models$ that corresponds
 to abstract interpretation.
It will be defined in detail in the upcoming
 Section~\ref{sec.model_abstractions}.
As a special case the fixed-point theorem also applies to standard
 semantics: Just replace $\modelsa$ by $\models$, ``model abstraction $(\calM,G)$'' by ``model $\calM$'' and ``abstract relation'' by ``relation''.

\begin{example}
Continuing the running example,
%
% Let $\exists X\, \psi$ be the running example,  where
% \[
% \exists X\,\psi \equiv \exists X\, ( X(2) \land \forall u (X(u) \impl X(u+u)) \land \neg X(3)),
% \]
% In Example \ref{ex.runningPhi} we saw that $\Phi_\psi$ is a $1$-tuple consisting of
% \[
% \phi(X,x) \equiv \exists u (x=2 \lor (X(u) \land x = u+u)).
% \]
% Thus $\mu(y) := [\lfp_X\, \Phi_\psi](y)$ is equivalent\fxnote{FIXME: theory axioms!?}
%  to $\exists z>0 (2^z = y)$.
Theorem \ref{thm.abstractFP} states that, for $\mu(y) := [\lfp_X\, \Phi_\psi](y)$,
\[
\exists X\, \psi \equiv (T\impl \mu(2)) \land \forall u (T\land \mu(u) \impl \mu(u+u)) \land (T\impl \neg \mu(3)).
\]
The first two clauses are true by the definition of $\mu$, hence $\exists X\, \psi$ is
 further equivalent to $T\impl \neg \mu(3)$.
\end{example}

\subsection{Model abstractions}\label{sec.model_abstractions}

% what do we want
In this section we define the semantics of {\em model abstractions}.
This is a generalization of the usual semantics of SOL[LFP] obtained by
 restricting the scope of second-order quantifiers to range only over
 certain relations and the domain in which the least fixed point
 of the lfp operator is computed.
The primary motivation for this semantics is that in many applications
 one is interested in speaking about the relations definable
 by certain formulas, see, e.g., Sections~\ref{sectionAffineSolutionProblem}
 and~\ref{sec.app_prog_verif}.
For example, we want to prove that a formula equation has a solution
 in a certain class of formulas.
We will cover this situation with the notion of {\em definable model abstractions}.
Since we also want to interpret the lfp operators, the relations need
 to form a complete lattice.

Model abstractions are defined via a Galois connection to the
 powerset lattice and thus they are the analogue of abstract
 interpretation for logical formulas.
Abstract interpretation has been introduced in~\cite{Cousot77Abstract} and is nowadays one of the most important techniques
in static analysis and software verification.
 For more background on Galois connections see \cite{Davey2002}.
The standard semantics is the special case where the Galois connection is the identity.

The idea of restricting the range of second-order quantifiers is reminiscent
 of Henkin semantics of second-order logic~\cite{Henkin50Completeness}.
Our abstract semantics differs from Henkin semantics in two
 respects: on the one hand we do not require, as Henkin semantics does, closure under definability which has the effect that there are
 model abstractions which are not Henkin structures.
On the other hand we require the lattice of sets to be complete which
 is necessary for interpreting lfp formulas in a suitable way.

\begin{definition}
	Let $\calA = (A,\subseteq)$ and $\calB = (B,\sqsubseteq)$ be two partially ordered sets.
	A \emph{Galois connection} between $\calA$ and $\calB$ consists of two functions
	$\alpha: A \rightarrow B$ and $\gamma: B \rightarrow A$, such that\ for all $X \in A$
	and $Y \in B$:
	\begin{align*}
		X \subseteq \gamma(Y) \quad \Leftrightarrow \quad \alpha(X) \sqsubseteq Y.
	\end{align*}
\end{definition}
From this condition it follows that $\alpha$ and $\gamma$ are monotone.

\begin{definition}
	Let $\calL$ be a language.
	A \emph{model abstraction} is a pair $(\calM,G)$, where $\calM$ is an $\calL$-structure and
	$G = (\mathcal{V}_k,\alpha_k,\gamma_k)_{k \in \Nat}$ is a sequence of triples, such that\ for all
	$k \in \Nat:$ $\mathcal{V}_k =(V_k,\sqsubseteq)$ is a lattice and
	$\alpha_k: \powset(M^k) \rightarrow V_k$ and $\gamma_k: V_k \rightarrow \powset(M^k)$ form a Galois-connection between $(\powset(M^k),\subseteq)$ and $\mathcal{V}_k$. We call $G$ the \emph{set domain} of $(\calM,G)$ and sets $R \sqsubseteq V_k$ \emph{abstract relations} of arity $k$.
\end{definition}

\begin{remark}
	Note that the completeness  of $(M^k,\subseteq)$ and the Galois-connection imply the completeness of the lattices $(V_k,\sqsubseteq)$.%\fxnote{add short proof or reference?}
\end{remark}

Model abstractions which are defined by a set of first-order formulas are of particular interest for
 applications in program verification.

\begin{definition}\label{def.definableModelabstr}
	Let $\calL$ be a language and $\calM$ be a model. Let $\calC$ be a set of first-order formulas such that\ the sets definable by formulas from $\calC$, i.e. $V_k = \{R \subseteq M^k ~|~ \exists \varphi \in \calC: \calM \models R(\mybar{x}) \liff \varphi(\mybar{x}) \}$, form a complete lattice $\calV_k = (V_k,\subseteq)$ for all $k \in \Nat$. Then there is a unique $\alpha_k$ such that $\alpha_k: \powset(M^k) \rightarrow V_k$ and $\id_k: V_k \rightarrow \powset(M^k)$ form a Galois-connection. We define $G_{\calC} = (\calV_k,\alpha_k,\id_k)_{k \in \Nat}$ and call $(\calM,G_{\calC})$ a \emph{definable model abstraction}.
\end{definition}

\begin{example}\label{ex.affone}
	Consider the language $\Laff = (0,1,+,(c)_{c \in \Rat})$ where the intended interpretation of the
	unary function symbol $c$ for $c\in \Rat$ is multiplication with $c$.
	Define the set domain $G_{\aff} = ((\Aff \Rat^k,\subseteq),\aff_k,\id_k)_{k \in \Nat}$, where $\Aff \Rat^k$ is the set of all affine
	subspaces of $\Rat^k$, $\aff_k$ maps every subset of $\Rat^k$ to its affine hull and $\id_k$ is the
	embedding of $\Aff \Rat^k$ in $\powset(\Rat^k)$.
	Then $(\Rat, G_{\aff})$ is a model abstraction. 
	Moreover, letting $\calC$ be the set of conjunctions of equations in $\Laff$, $(\Rat, G_{\aff})$
	turns out to be the definable model abstraction $(\Rat,G_\calC)$.
\end{example}

The intention of the notion of definable model abstraction is seen most clearly in
 Theorem~\ref{thm.hoareVcEquiv}, which shows that the truth of the verification condition
 of a partial correctness assertion in $(\calM, G_\calC)$ is equivalent to its provability in
 the Hoare calculus with invariants restricted to formulas from $\calC$.

We will now introduce the abstract semantics of SO[LFP] formulas by defining a satisfaction relation
 $(\calM, G) \modelsa \varphi$.
The crucial difference between $\modelsa$ and standard Tarski semantics $\models$ will be that
second-order quantifiers and the least fixed-point operator will not be interpreted in the power set
of the domain but in $G$ instead (for the appropriate arity).
\begin{definition}
	The defining clauses for first-order atoms, propositional connectives, and first-order
	quantifiers for $\modelsa$ are identical to those for $\models$. 
	For formulas of the form $\exists X\, \psi$, where $X$ is a $k$-ary predicate variable, we define
	\[
	(\calM,G),\theta \modelsa \exists X \psi \quad \Leftrightarrow \quad \exists R \in V_k: (\calM,G),\theta[X := \gamma_k(R)] \modelsa \psi,
	\]
	and analogously for formulas of the form $\forall X \psi$.
	The semantics of the $\lfp$-operator is defined as follows. Let $(\calM,G)$ be a model abstraction and $X_1,\ldots,X_n$ be predicate variables with $X_i$ having arity $k_i$. Let
	$\Phi = (\varphi_i(X_1,\ldots,X_n,\mybar{u_i}))_{i=1}^n$ be an $n$-tuple of first-order formulas
	such that\ $|\mybar{u_i}| = k_i$ and $X_1,\ldots,X_n$ occur only positively in $\Phi$.
	Define
	\begin{align*}
		F_i^{\#}: ~V_{k_1} \times \cdots \times V_{k_n} &\rightarrow V_{k_i}\\
		(Y_1, \ldots ,Y_n) &\mapsto \alpha_{k_i} \circ (F_{\Phi}^{\calM})_i(\gamma_{k_1}(Y_1), \ldots ,\gamma_{k_n}(Y_n)),
	\end{align*}
	and $F_{\Phi}^{\#} = (F_1^{\#},\ldots,F_n^{\#})$.
	Let $\mybar{a} \in \calM$, then
	\[
	(\calM,G) \modelsa [\lfp_{X_i} ~\Phi](\mybar{a}) \quad \Leftrightarrow \quad \mybar{a} \in \gamma_{k_i}( \lfp(F_{\Phi}^{\#})_i).
	\]

	Let $R$ be an abstract relation of arity $k$. Then we write $(\calM,G),\theta \modelsa \varphi(R)$ as abbreviation for $(\calM,G),\theta[X := \gamma_k(R)] \modelsa \varphi(X).$
\end{definition}
Note that the semantics of the least fixed-point operator is well-defined:
We already know that $F_{\Phi}$ is a monotone $n$-ary operator and $\alpha_{k_i}$ and $\gamma_{k_i}$ are monotone
as they form a Galois connection for $i \in \{1,\ldots,n\}$.
Thus $F_i^{\#}$ is monotone for all $i \in \{1,\ldots,n\}$ and therefore $F_{\Phi}^{\#}$ is monotone as
well. As $\mathcal{V}_k$ is a complete lattice for all $k \in \Nat$, we can use the Knaster-Tarski theorem
to obtain the least fixed point of $F_{\Phi}^{\#}$.

Note that every classical model defines, by choosing the identity for
 $\alpha_k$ and $\gamma_k$, a model abstraction, that satisfies the
 same formulas.
Hence $\modelsa \varphi \Rightarrow~ \models \varphi$ for every SO[LFP]
 formula $\varphi$, and for a formula equation $\exists \mybar{X}\, \varphi$ we even have
 $(\calM, G) \modelsa \exists \mybar{X}\, \varphi~ \Rightarrow~ \calM \models \exists \mybar{X}\, \varphi$.
 The converse is not true as the following example shows,  yet for first-order
 formulas $\varphi$ we have $\modelsa \varphi \Leftrightarrow~ \models \varphi$.
In particular, if $\varphi_1$ and $\varphi_2$ are logically equivalent
 first-order formulas, we also have
 $\modelsa \varphi_1 \leftrightarrow \varphi_2$.
\begin{example}
We have observed in Example~\ref{ex.feq} that the formula equation $\Phi$ defined there
 is valid, i.e., $\models \Phi$.
However, we do not have $\modelsa \Phi$.
To see this, consider the model abstraction $(\Nat, G)$ where $\Nat$ is considered
 as structure in the language $\calL = \{ 0, s \}$, $G = (\mathcal{V}_k,\alpha_k,\gamma_k)_{k \in \Nat}$ with $\mathcal{V}_k = (V_k,\sqsubseteq)$ where $V_k = \{ \emptyset, \Nat^k \}$,
 $\alpha_k = \id \restriction V_k$ and $\gamma_k = \id$.
Then $(\Nat,G) \nmodelsa \Phi$ because neither $\emptyset$ nor $\Nat$ is a solution
 for the formula equation $\Phi$.
\end{example}

\subsection{Proof of the fixed-point theorem}
The proof of the abstract fixed-point theorem relies on a generalisation of a result by  Nonnengart
and Sza{\l}as~\cite{Nonnengart98Fixpoint}, which is in turn a generalisation of Ackermann's
Lemma~\cite{Ackermann35Untersuchungen}.
Ackermann`s Lemma states that for every formula $\varphi$ of the form
$\exists X(\forall \mybar{x}(\alpha(\mybar{x})\rightarrow X(\mybar{x}))\wedge \beta(X))$,
where $\alpha$ and $\beta$ are first-order and $X$ occurs only negatively in $\beta$, the formula $\beta(\alpha(\mybar{x}))$ is equivalent to $\varphi$.
Nonnengart and Sza{\l}as' generalisation allows positive occurrences of $X$ in $\alpha(X,\mybar{x})$.
In order to deal with this more general case, they use a fixed-point operator.
The equivalent formula becomes the FO[LFP] formula $\beta([\lfp_X\, \alpha(X,\mybar{x})])$,
 i.e., they prove
\[\models \exists X(\forall \mybar{x}(\alpha(X,\mybar{x})\rightarrow X(\mybar{x}))\wedge \beta(X)) \leftrightarrow \beta([\lfp_X\, \alpha(X,\mybar{x})]).\]
We further generalise that allowing an arbitrary number of predicate variables by using a simultaneous least fixed-point, extending the scope of the result from standard semantics to abstract semantics, and adding property (2) below on the minimality of the solution.
This lemma will be the main tool in the proof of the abstract fixed-point theorems.
\begin{lemma}\label{lem.abstractNonnengart}
	Let $X_1,...,X_n$ be predicate variables. Let $\beta(X_1,...,X_n)$ be a first-order formula and $\Phi = (\alpha_i(X_1,...,X_n,\mybar{x_i},))_{i=1}^n$ be an $n$-tuple of first-order formulas such that $|\mybar{x_i}|$ equals the arity of $X_i$ for $i \in \{1,...,n\}$. \footnote{Note that there could also be free variables in $\alpha_1,...,\alpha_n,\beta$, that are not stated explicitly.} If $X_1,...,X_n$ occur only positively in $\alpha_1,...,\alpha_n$ and only negatively in $\beta$, then 
	\begin{enumerate}
		\item $\modelsa \exists \mybar{X} \left(\bigwedge_{i=1}^n \forall \mybar{x_i} \left(\alpha_i(X_1,...,X_n,\mybar{x_i}) \impl X_i(\mybar{x_i})\right) \wedge \beta(X_1,...,X_n) \right) \leftrightarrow \beta([\lfp_{X_1} ~\Phi],...,[\lfp_{X_n} ~\Phi])$,
		\item If  $(\calM,G) \modelsa \bigwedge_{i=1}^n \forall \mybar{x_i} \left(\alpha_i( R_1,...,R_n,\mybar{x_i}) \impl R_i(\mybar{x_i})\right)$ for some model abstraction $(\calM,G)$ and abstract relations $R_1,...,R_n$, then $(\calM,G) \modelsa \bigwedge_{i=1}^n \forall \mybar{x_i} ([\lfp_{X_i} ~\Phi](\mybar{x_i}) \impl R_i(\mybar{x_i}))$.
	\end{enumerate}
\end{lemma}
\begin{proof}
	1. $``\rightarrow"$: Let $(\calM,G)$ be a model abstraction and $\theta$ an environment such that
	\[
	(\calM,G),\theta \modelsa \exists \mybar{X} \bigwedge_{i=1}^n \forall \mybar{x_i} \left(\alpha_i(X_1,...,X_n,\mybar{x_i}) \impl X_i(\mybar{x_i})\right) \wedge \beta(X_1,...,X_n).
	\]
	Hence there exist abstract relations $S_1,...,S_n$ and an environment $\theta' := \theta[X_1 := \gamma_{k_1}(S_1),...,X_n := \gamma_{k_n}(S_n)]$  such that 
	\begin{align}\label{abstractNonnengartproof}
		(\calM,G),\theta' \modelsa \bigwedge_{i=1}^n \forall \mybar{x_i} \left(\alpha_i(X_1,...,X_n,\mybar{x_i}) \impl X_i(\mybar{x_i})\right) \wedge \beta(X_1,...,X_n).
	\end{align}
Then $(F_{\Phi})_i(\gamma_{k_1}(S_1),...,\gamma_{k_n}(S_n)) \subseteq \gamma_{k_i}(S_i)$ for all $i \in \{1,...,n\}$ and as $\alpha_k,\gamma_k$ form a Galois connection for every $k$ we have $\alpha_{k_i} \circ (F_{\Phi})_i(\gamma_{k_1}(S_1),...,\gamma_{k_n}(S_n)) \sqsubseteq S_i$ for all $i \in \{1,...,n\}$. Thus $(S_1,...,S_n)$ is a fixed point of $F_{\Phi}^{\#}$. Using the monotonicity of $\gamma_k$ for every $k$, this implies
\begin{align}\label{abstactNonnengartproofPart2}
	(\calM,G),\theta' \modelsa \bigwedge_{i=1}^n \forall \mybar{x_i} ([\lfp_{X_i} ~\Phi](\mybar{x_i}) \impl X_i(\mybar{x_i})).
\end{align}

	Note that the FO[LFP]-formulas are well defined, as $X_1,...,X_n$ occur only positively in $\alpha_1,...,\alpha_n$. From (\ref{abstractNonnengartproof}) we also obtain $(\calM,G),\theta' \modelsa \beta(X_1,...,X_n)$ and as $X_1,...,X_n$ occur only negatively in $\beta$ we may use Lemma \ref{lem.ImplPositiv} to obtain
	\[
	(\calM,G),\theta' \modelsa \beta([\lfp_{X_1} ~\Phi],...,[\lfp_{X_n} ~\Phi]).
	\]
	This formula does not depend on $X_1,...,X_n$, hence$
	(\calM,G),\theta \modelsa \beta([\lfp_{X_1} ~\Phi],...,[\lfp_{X_n} ~\Phi]).
	$\\
	$``\leftarrow"$:  For the other direction let $(\calM,G)$ be a model abstraction and $\theta$ an environment such that
	\[
	(\calM,G),\theta \modelsa \beta([\lfp_{X_1} ~\Phi],...,[\lfp_{X_n} ~\Phi]).
	\]
	Let $S_i \in V_{k_i}$ be the abstract relations defined by $[\lfp_{X_i} ~\Phi]$ for $i \in \{1,...,n\}$. $(S_1,...,S_n)$ is a fixed point of $F_{\Phi}^{\#}$, therefore by the definition of $F_{\Phi}^{\#}$ we have 
	\[
	\alpha_{k_i} \circ (F_{\Phi})_i(\gamma_{k_1}(S_1),...,\gamma_{k_n}(S_n)) \sqsubseteq S_i, \quad \text{for every}~ i \in \{1,...,n\},
	\]
	and as $\alpha_k,\gamma_k$ form a Galois connection for all $k$ we get 
	\[
	(F_{\Phi})_i(\gamma_{k_1}(S_1),...,\gamma_{k_n}(S_n)) \subseteq \gamma_{k_i}(S_i), \quad \text{for every}~ i \in \{1,...,n\}.
	\]
	Thus $(\gamma_{k_1}(S_1),...,\gamma_{k_n}(S_n))$ is a fixed point of $F_{\Phi}$ and by defining $\theta' := \theta[X_1 := \gamma_{k_1}(S_1),...,X_n := \gamma_{k_n}(S_n)]$, we obtain
	\[
	(\calM,G),\theta' \modelsa \bigwedge_{i=1}^n \forall \mybar{x_i} \left(\alpha_i(X_1,...,X_n,\mybar{x_i}) \impl X_i(\mybar{x_i})\right).
	\]
	By assumption $(\calM,G),\theta' \modelsa \beta(X_1,...,X_n)$ and therefore we conclude
	\[
	(\calM,G),\theta \modelsa \exists \mybar{X} \bigwedge_{i=1}^n \forall \mybar{x_i} \left(\alpha_i(X_1,...,X_n,\mybar{x_i}) \impl X_i(\mybar{x_i})\right) \wedge \beta(X_1,...,X_n).
	\]
	We get 2. directly from (\ref{abstactNonnengartproofPart2}), if we choose $S_1,...,S_n$  to be the abstract relations $R_1,...,R_n$.
\end{proof}

We are now able to prove the abstract fixed-point theorem, the main result of this paper. For clarity, we restate the theorem.
\newcounter{theoremRem}
\setcounter{theoremRem}{\thetheorem}
\setcounter{theorem}{\getrefnumber{thm.abstractFP}}
\addtocounter{theorem}{-1}
\begin{theorem}[Abstract Horn fixed-point theorem]
	Let $\exists \mybar{X} \psi$ be a Horn formula equation and $\mu_j := [\lfp_{X_j} ~\Phi_{\psi}]$ for $j \in \{1,...,n\}$, then:
	\begin{enumerate}
		\item $\modelsa \exists \mybar{X}  ~\psi \leftrightarrow \psi[\mybar{X}\backslash \mybar{\mu}]$ and
		\item if $(\calM,G) \modelsa \psi[\mybar{X}\backslash \mybar{R}]$ for some model abstraction $(\calM,G)$ and abstract relations $R_1,...,R_n$, then $(\calM,G) \modelsa \bigwedge_{j=1}^n (\mu_j \rightarrow R_j)$.
	\end{enumerate}
\end{theorem}
\begin{proof}%[Proof of Theorem \ref{thm.abstractFP}]
	$``\rightarrow"$: Assume that $(\calM,G) \modelsa \exists \mybar{X} \psi$. We first note that in the transformation of Lemma \ref{lem.Hornformula.split} there does not occur any second order quantifiers nor fixed point operators. Thus it also holds for abstract models and we obtain
	\[
	(\calM,G) \modelsa \exists \mybar{X} \bigwedge_{j=1}^n \forall \mybar{x_j} (\varphi_j(X_1,...,X_n,\mybar{x_j}) \impl X_j(\mybar{x_j})) \wedge \forall \mybar{y} \bigwedge_{C \in E} C(X_1,...,X_n,\mybar{y}),
	\]
	where $X_1,...,X_n$ occur only positively in $\varphi_1,...,\varphi_n$ and only negatively in the clauses $C$. Hence we can apply Lemma \ref{lem.abstractNonnengart}/1 to obtain
	\[
	(\calM,G) \modelsa \forall \mybar{y} \bigwedge_{C \in E} C(\mu_1,...,\mu_n,\mybar{y})
	\]
	By construction, the FO[LFP] formulas $\mu_1,...,\mu_n$ also satisfy the base and induction clauses, i.e. for any base or induction clause $C(X_1,...,X_n,\mybar{y})$ it holds $(\calM,G) \modelsa \forall \mybar{y}\, C(\mu_1,...,\mu_n,\mybar{y})$. Therefore it follows that
	\[
	(\calM,G) \modelsa \psi[X_1\backslash \mu_1,...,X_n\backslash\mu_n]
	\]
	The other direction ~$``\leftarrow"$~ is immediate.
	The second part of the theorem follows directly from Lemma \ref{lem.abstractNonnengart}/2.
\end{proof}
\setcounter{theorem}{\thetheoremRem}

We also present the fixed-point theorem for classical semantics. This has been proven directly in \cite{Kloibhofer20Fixed}, here it is a corollary of the abstract fixed-point theorem as every classical model is also an abstract model if the abstraction and concretisation functions are chosen as the identity functions.
\begin{corollary}\label{cor.HornFP}
	Let $\exists \mybar{X} \psi$ be a Horn formula equation and $\mu_j := [\lfp_{X_j} ~\Phi_{\psi}]$
	for $j \in \{1,\ldots,n\}$, then
	\begin{enumerate}
		\item $\models \exists \mybar{X}\,\psi \liff \psi\unsubst{\mybar{X}}{\mybar{\mu}}$ and
		\item if  $\calM \models \psi \unsubst{\mybar{X}}{\mybar{R}}$
		for some structure $\calM$ and relations $R_1,\ldots,R_n$ in $\calM$, then
		$\calM \models \Land_{j=1}^n ( \mu_j \impl R_j )$.
	\end{enumerate}
\end{corollary}

\begin{remark}
Note that Corollary~\ref{cor.HornFP} applies equally to constraints being FO[LFP]-formulas.
It therefore shows that FO[LFP], in contrast to first-order logic, has the property of being
 closed under solving Horn formula equations.
It thus shows that in FO[LFP] validity and solvability of Horn formula equations coincide.
This is in contrast to formula equations in first-order logic, cf.~Example~\ref{ex.feq}.
\end{remark}

\begin{remark}
One of the most prominent approaches to solving a SOQE problem $\exists \mybar{X}\, \varphi$ is to
 solve the corresponding WSOQE problem to obtain formulas $\mybar{\chi}$
 such that\ $\models \exists \mybar{X}\, \varphi \liff \varphi\unsubst{\mybar{X}}{\mybar{\chi}}$.
For example, it underlies the family of DLS algorithms: the original DLS~\cite{Doherty97Computing},
 $\text{DLS}^*$~\cite{Doherty98General,Nonnengart98Fixpoint}, and $\text{DLS}'$~\cite{Eberhard17Boolean}.
In this context part (1) of Corollary~\ref{cor.HornFP} can be understood as yielding a solution
 to the WSOQE-, and hence the SOQE-problem for Horn formula equations in FO[LFP].
\end{remark}

\subsection{The Dual and Linear-Horn fixed-point theorems}\label{sec.dualLinearfpthm}

We now turn our interest to dual-Horn and linear-Horn formula equations. Recall that a constrained dual-Horn clause consists of at most one positive predicate variable, in contrast to constrained Horn clauses, where at most one negative predicate variables occurs. Thus dual-Horn formula equations are dual to Horn formula equations.

For a formula $\psi$ we define $\psi^D$ as $\psi\sop \sel{X_1}{\neg X_1}, \ldots, \sel{X_n}{\neg X_n} \scl$,
 where $X_1,\ldots,X_n$ are all predicate variables occurring in $\psi$.
Note that $\psi \equiv (\psi^D)^D$ for all formulas $\psi$.
Moreover, note that $\models \exists \mybar{X}\, \psi \liff \exists \mybar{X}\, \psi^D$.
If $\exists \mybar{X}\, \psi$ is a Horn formula equation, then $\exists \mybar{X}\, \psi^D$ is logically
 equivalent to a dual-Horn formula equation and if $\exists \mybar{X}\, \varphi$ is a dual-Horn formula
 equation, then $\exists \mybar{X}\, \varphi^D$ is logically equivalent to a Horn formula equation.
Note that dualisation of a (dual) Horn formula equation interchanges (B)- and (E)-clauses.
\begin{example}
Consider the constrained dual-Horn clauses in the language of arithmetic
\[
\psi \quad\equiv\quad X(1) \land ( X(n) \impl X(n+1) \lor X(n+2)) \land \neg X(4).
\]
The dualisation of $\psi$ is
\[
\psi^D \quad\equiv\quad \neg X(1) \land ( \neg X(n) \impl \neg X(n+1) \lor \neg X(n+2)) \land X(4)
\]
which is logically equivalent to the constrained Horn clauses
\[
\neg X(1) \land ( X(n+1) \land X(n+2) \impl  X(n)) \land X(4)
\]
\end{example}

If $\exists \mybar{X} \psi$ is a dual-Horn formula equation, then $\exists \mybar{X} \psi^D$ is equivalent to a Horn formula equation. Now let $\mybar{\mu}$ be the least solution of $\exists \mybar{X} \psi^D$, then $\mybar{\neg \mu}$ is the greatest solution of $\exists \mybar{X} \psi$. Thus the dual-Horn fixed-point theorem follows from the Horn case in
 Theorem~\ref{thm.abstractFP}.

\begin{theorem}[Abstract dual-Horn fixed-point theorem]\label{thm.abstractDualFP}
	Let $\exists \mybar{X} \psi$ be a dual-Horn formula equation and $\nu_j := \neg [\lfp_{X_j} ~\Phi_{\psi^D}]$ for $j \in \{1,...,n\}$, then:
	\begin{enumerate}
		\item $\modelsa \exists \mybar{X}  ~\psi \leftrightarrow \psi[\mybar{X}\backslash \mybar{\nu}]$ and
		\item if $(\calM,G) \modelsa \psi[\mybar{X}\backslash \mybar{R}]$ for some model abstraction $(\calM,G)$ and abstract relations $R_1,...,R_n$, then $(\calM,G) \modelsa \bigwedge_{j=1}^n (R_j \rightarrow \nu_j)$.
	\end{enumerate}
\end{theorem}

\begin{proof}
Let $\mu_j := [\lfp_{X_j} ~\Phi_{\psi^D}]$ for $j = 1,\ldots,n$.
For 1. note that, since $\exists \mybar{X}\, \psi$ is a dual Horn formula
 equation, $\exists \mybar{X}\, \psi^D$ is logically equivalent to a Horn formula equation. % FIXME: do we identify formula equations modulo logical equivalence?
An application of Theorem \ref{thm.abstractFP}.1. yields
 $\modelsa \exists \mybar{X}\, \psi^D \liff \psi^D \unsubst{\mybar{X}}{\mybar{\mu}}$.
Since $\psi^D\unsubst{\mybar{X}}{\mybar{\mu}}$ is syntactically equal to $\psi\unsubst{\mybar{X}}{\mybar{\nu}}$
 we obtain $\modelsa \exists \mybar{X}\, \psi \liff \exists \mybar{X}\psi^D \liff \psi^D\unsubst{\mybar{X}}{\mybar{\mu}} \liff \psi\unsubst{\mybar{X}}{\mybar{\nu}}$.
 
For 2.\ assume that $(\calM,G) \modelsa \psi[X_1\backslash R_1,...,X_n\backslash R_n]$ for some model abstraction $(\calM,G)$ and abstract relations $R_1,...,R_n$.
Then  $(\calM,G) \modelsa \psi^D[X_1\backslash \neg R_1,...,X_n\backslash\neg R_n]$, so, by Theorem \ref{thm.abstractFP}/2., $\calM \models \Land_{j=1}^n ( \mu_j \impl \neg R_j)$
 which yields $(\calM,G) \modelsa \Land_{j=1}^n ( R_j \rightarrow \nu_j )$ by contraposition.
\end{proof}

Note that the operator induced by $\Phi_{\psi^D}$ is not the dual operator of the one induced by $\Phi_\psi$
 in the sense of~\cite{Fritz01Fixed} because $\Phi_{\psi^D}$ is not the (pointwise) negation of
 $\Phi_\psi$.
Therefore $\nu$ is not the greatest fixed point of $\Phi_\psi$.

As a constrained linear-Horn clause is both a constrained Horn and a constrained dual-Horn clause we can combine the two fixed-point theorems and we obtain for the case of linear-Horn formula equations:
\begin{theorem}[Abstract linear-Horn fixed-point theorem]\label{thm.abstractLinearFP}
	Let $\exists \mybar{X} \psi$ be a linear-Horn formula equation, $\mu_j := [\lfp_{X_j} ~\Phi_\psi]$ and $\nu_j := \neg [\lfp_{X_j} ~\Phi_{\psi^D}]$ for $j \in \{1,...,n\}$, then:
	\begin{enumerate}
		\item $\modelsa \exists \mybar{X}  ~\psi \leftrightarrow \psi[\mybar{X}\backslash \mybar{\mu}]$ and $\modelsa \exists \mybar{X}  ~\psi \leftrightarrow \psi[\mybar{X}\backslash \mybar{\nu}]$ and
		\item if $(\calM,G) \modelsa \psi[\mybar{X}\backslash \mybar{R}]$ for some model abstraction $(\calM,G)$ and abstract relations $R_1,...,R_n$, then $(\calM,G) \modelsa \bigwedge_{j=1}^n (\mu_j \rightarrow R_j \land R_j \rightarrow \nu_j)$.
	\end{enumerate}
\end{theorem}

An \emph{interpolant} of two variable-free tuples of FO[LFP] formulas $\mybar{\mu}$ and
 $\mybar{\nu}$ in the same first-order language is a tuple of first-order formulas
 $\mybar{\chi}$ such that
 $\models \bigwedge_{j=1}^n (\mu_j \rightarrow \chi_j \land \chi_j \rightarrow \nu_j)$.
Theorem \ref{thm.abstractLinearFP}/2 states that every solution of the linear-Horn formula equation is an interpolant of $\mybar{\mu}$ and $\mybar{\nu}$. The converse is only true if we add another assumption:
\begin{corollary}\label{cor.linHorn_interpol}
	Let $\exists \mybar{X} \psi$ be a linear-Horn formula equation, $\mu_j := [\lfp_{X_j} ~\Phi_\psi]$ and $\nu_j := \neg [\lfp_{X_j} ~\Phi_{\psi^D}]$ for $j \in \{1,...,n\}$, then:
	\begin{enumerate}
		\item Every solution of $\exists \mybar{X} \psi$ is an interpolant of $\mybar{\mu}$ and $\mybar{\nu}$,
		\item an interpolant $\mybar{\chi}$ is a solution of $\exists \mybar{X} \psi$ iff $\mybar{\chi}$ satisfies all induction clauses.
	\end{enumerate}
\end{corollary}
\begin{proof}
	1. is an immediate consequence of Theorem \ref{thm.abstractLinearFP}/2. For 2. we first note that $\mybar{X}$ occur only positively in the base clauses and only negatively in the end clauses. Thus we can use Lemma \ref{lem.ImplPositiv} to obtain that the interpolant $\mybar{\chi}$ satisfies all base and end clauses, hence $\mybar{\chi}$ is a solution iff it also satisfies all induction clauses.
\end{proof}
Interpolation is an important technique for solving constrained Horn clauses, see,
 e.g.,~\cite{McMillan12Solving}.
The above result provides a theoretical connection between interpolation and verification.
The relationship between interpolation and Horn clauses has also been
 studied by encoding interpolation problems with a language condition on the constant
 symbols as Horn clause sets~\cite{Ruemmer13Classifying,Gupta14Generalised}.

\section{Decidability of the affine solution problem}\label{sectionAffineSolutionProblem}

As an application of the abstract fixed-point theorem we take a look at the affine
 solution problem, which is shown to be decidable in \cite{Hetzl20Decidability}.
This result is based on a generalisation of Karr's algorithm~\cite{Karr76Affine}
 to non-Horn formula equations.
It is proved by computing a fixed point similarly to how we did in Section~\ref{sec.afpthm}.
The main difference is that the fixed point is not computed on the logical level, but in
 the lattice of affine subspaces of $\Rat^n$.
The abstract fixed point theorem shown in Section~\ref{sec.afpthm} applies to this problem
 and drastically shortens the proof of decidability.

The setting of the affine solution problem is the language $\Laff = (0,1,+,(c)_{c \in \Rat})$
 as discussed in Example~\ref{ex.affone}.
As we are only working in $\Rat$ we can assume without loss of generality, that every term $t(x_1,...,x_n)$ is of the form $c_0 + \sum_{i=1}^m c_i x_i$ and every atomic formula $A(x_1,...,x_n)$ is of the form $c_0 + \sum_{i=1}^m c_i x_i = 0$. We call such atomic formulas \emph{linear equations} and conjunctions of linear equations \emph{linear equation systems}.
It is well known that linear equation systems define affine subspaces of $\Rat^n$.

\begin{definition}
	An \emph{affine formula equation} is a formula equation of the form $\exists \mybar{X} ~\psi$, where $\psi$ is a quantifier-free first-order formula in $\Laff\cup \{X_1,...,X_n\}$. 
\end{definition}

\begin{definition}
	The affine solution problem is the set of affine formula equations that have solutions in the class of linear equation systems in $\Rat$, i.e. the set 
	\[
	\{\exists \mybar{X} \psi ~|~ \text{There exist linear equation systems}~ F_1,...,F_n ~\text{such that}~ \Rat \models \psi[X_1\backslash F_1,...,X_n\backslash F_n]\}.
	\]
\end{definition}

Recall from Example \ref{ex.affone} the set domain  $G_{\aff} = ((\Aff \Rat^k,\subseteq),\aff_k,\id_k)_{k \in \Nat}$, where $\Aff \Rat^k$ is the set of all affine subspaces of $\Rat^k$, $\aff_k$ maps every subset of $\Rat^k$ to its affine hull and $\id_k$ is the embedding of $\Aff \Rat^k$ in $\powset(\Rat^k)$. Then $(\Rat, G_{\aff})$ is a model abstraction. $G_{\aff}$ is defined such that an affine formula equation $\exists \mybar{X} ~\psi$ is in the affine solution problem iff $(\Rat,G_{\aff}) \modelsa \exists \mybar{X} ~\psi$. Thus we can use the abstract fixed-point theorem.

\begin{theorem}\label{thm.affinesolprob}
	The affine solution problem is decidable.
\end{theorem}
\begin{proof}
	As in~\cite{Hetzl20Decidability} we reduce the solvability of a formula equation $\exists \mybar{X} \psi$ to the
	solvability of one of its finitely many projections, which are Horn formula equations.
	Let $\exists \mybar{X} \varphi$ be one of them, then we may apply Theorem~\ref{thm.abstractFP}. This yields a tuple of FO[LFP]-formulas $\mybar{\mu}$ such that\ 
	$(\Rat,G_{\aff}) \modelsa \exists \mybar{X}\, \varphi \leftrightarrow \varphi\unsubst{\mybar{X}}{\mybar{\mu}}$.
	Since all lattices $\Aff \Rat^k$ have finite height we can compute fixed-point free formulas $\mybar{\chi}$ equivalent to $\mybar{\mu}$ from $\mybar{\mu}$. 
	
	For notational simplicity we describe how this is achieved in the case of one predicate variable $X$. Let $\mu = [\lfp_X~ \Phi_\varphi]$ be a FO[LFP]-formula. As in Section \ref{sec.prelim.approximation}, $\mu$ can be approximated by the sequence of sets in $V_k$
\begin{align*}
	S^0 &:= \varnothing,\\
	S^{l+1} &:= F_{\Phi}^{\#}(S^l).
\end{align*}
These sets describe increasing affine subspaces of $M^k$, thus, as the lattice $\Aff \Rat^k$ has height $k$, it holds $F_{\Phi}^{\#}(S^k) = S^k$ and therefore $\mybar{a} \in S^k \Leftrightarrow (\Rat,G_{\aff}) \modelsa \mu(\mybar{a})$.
In the actual computation, the affine subspace $S^i$ is represented by a linear equation
 system, i.e., a formula $\sigma_i$.
The computation of $F_{\Phi}^{\#}$ is based on the fact that, given finite representations
 of affine spaces $\calA$ and $\calB$ and of an affine transformation
 $\mathcal{T}$, it is possible to compute finite representations of $\calA \cap \calB$,
 $\calA \cup \calB$, $\calT(\calA)$, and $\calT^{-1}(\calA)$.
This is an elementary result of affine geometry and linear algebra,
 see~\cite{Snapper71Metric}.
% It remains to compute a first-order formula $\sigma_k$, which describes $S^l$. Similarly to Subsection \ref{sec.prelim.approximation} we define the formulas
% \begin{align*}
% 	\sigma_0 &:= \bot,\\
% 	\sigma_{l+1} &:= f_{\aff}(\Phi_\psi(\sigma_l)),
% \end{align*}
% where $f_{\aff}$ computes, given as an input a first-order formula, a first-order formula describing its affine hull. It is easy to see that $f_{\aff}$ is computable.
% By the definition of $F_{\Phi}^{\#}$ it holds $\mybar{a} \in S^l \Leftrightarrow \Rat \models \sigma_l$.
Hence $\sigma_k$ is the desired first-order formula such that $(\Rat,G_{\aff}) \modelsa \mu(\mybar{a}) \Leftrightarrow \Rat \models \sigma_k(\mybar{x})$.
	
	Therefore $(\Rat,G_{\aff}) \modelsa \exists \mybar{X}\, \varphi \leftrightarrow \varphi\unsubst{\mybar{X}}{\mybar{\chi}}$.
	Now $\varphi\unsubst{\mybar{X}}{\mybar{\chi}}$ is a first-order formula and hence $(\Rat,G_{\aff}) \modelsa \varphi\unsubst{\mybar{X}}{\mybar{\chi}}$
	iff $\Rat \models \varphi\unsubst{\mybar{X}}{\mybar{\chi}}$.
	The latter statement can be checked by a decision procedure for linear arithmetic.
\end{proof}

Karr's algorithm was extended to the computation of polynomial invariants
 in~\cite{RodriguezCarbonell07Automatic,RodriguezCarbonell07Generating}, see
 also~\cite{Hrushovski18Polynomial,Hrushovski23Strongest}.
The authors do no know whether the analogue of Theorem~\ref{thm.affinesolprob} for the case of
 polynomials is true.
If it is, a different proof strategy will be necessary because the reduction to projections
 is not possible in the case of polynomials.

\section{Applications to program verification}\label{sec.app_prog_verif}

In this section we will describe some direct applications of our fixed-point theorems to
 the foundations of program verification.
As an exemplary framework we will consider the Hoare calculus for a simple
 imperative programming language as in~\cite{Winskel93Formal}.
The two main goals of this section are:
\begin{enumerate}
\item We show that proving a partial correctness assertion amounts to
 solving a linear Horn formula equation in Theorem~\ref{thm.vcCompleteness}.
\begin{enumerate}
\item This allows to express the partial correctness of a program as an FO[LFP] formula,
 see Corollary~\ref{HoareVCLinearCorollary}.
\item We show that the canonical solutions of this linear Horn formula equation correspond
 to the least precondition and strongest postcondition in Theorems~\ref{thm.wpLfpNu}
 and~\ref{thm.spLfpMu}.
\end{enumerate}

\item We give an alternative proof of completeness of Hoare logic based on our abstract
 fixed-point theorem in Theorem~\ref{thm.hoareVcEquiv} which has the following advantages:
\begin{enumerate}
\item It does not need the expressivity hypothesis of the standard proof for $\Int$.
\item It accommodates abstract interpretation in the form of a restriction on the syntactic
 form of the loop invariants.
\end{enumerate}
\end{enumerate}

\subsection{Hoare triples}

In the study of program verification it is quite common to only work in the language of arithmetic and the structure $\Int$. This may disguise the importance of the expressivity of $\Int$. As we are aiming for a framework where the usage of FO[LFP]-formulas replaces the encoding of finite sequences we work in a more general setting.

\begin{definition}
	Let $\calL$ be a language. The set of programs in $\calL$ is defined by
	\begin{align*}
		p ::= \textbf{skip} ~|~ x := t ~|~ p_0;p_1 ~|~ \textbf{if}~ B~\textbf{then}~ p_0 ~\textbf{else} ~p_1 ~|~ \textbf{while}~ B ~\textbf{do}~ p_0,
	\end{align*}
	where $t$ is an $\calL$-term, $B$ a quantifier-free first-order formula in $\calL$
	and $x$ is a program variable.
\end{definition}

\begin{definition}
	A \emph{Hoare triple} is a triple $(\varphi,p,\psi)$ consisting of a program $p$ and two first-order formulas $\varphi$ and $\psi$. This is traditionally denoted as $\hoare{\varphi}{p}{\psi}$. 
\end{definition}

\begin{definition}
	Let $\calL$ be a language and $\calM$ be a structure with domain $M$. Define $\mathrm{Var} = \{x_0,x_1,...\}$ to be the set of variables which may occur in a program. A function $\sigma: \mathrm{Var} \rightarrow M$ is called a \emph{state}\footnote{Later we will talk of formulas, in which the variables of a program occur. In this sense a state may also be seen as an environment.}. We denote the set of all states with $\Sigma$. For $m \in M$ we write $\sigma[x_j \rightarrow m]$ for the unique state $\sigma'$ such that $\sigma'(x_j) = m$ and $\sigma'(x_i)= \sigma(x_i)$ for $i \neq j$.
\end{definition}

\begin{definition}[Denotational semantics]
	For every program $p$ we define a relation $C(p)$ on $\Sigma \times \Sigma$ by structural induction. 
	\begin{alignat*}{2}
		C(\hskipit) =&~ \{(\sigma,\sigma) ~|~ \sigma \in \Sigma \} \\
		\\
		C(x_j := t) =&~ \{(\sigma, \sigma[x_j \rightarrow m]) ~|~ \sigma \in \Sigma,~ m \in M ~\text{and}~ \calM,\sigma \models t = m\} \\ 
		\\
		C(p_0;p_1) =&~ C(p_1) \circ C(p_0) \\
		\\
		C(\hif B \hthen p_0 \helse p_1) =&~\{(\sigma,\sigma') ~|~ \calM,\sigma \models B ~\text{and}~ (\sigma,\sigma') \in C(p_0) \} ~\cup \\
		&~\{(\sigma,\sigma') ~|~ \calM,\sigma \models \neg B~\text{and}~ (\sigma,\sigma') \in C(p_1) \}\\
		\\
		C(\hwhile B \hdo p_0) =&~ \lfp(\Gamma),
	\end{alignat*}
	where $ \Gamma$ is an operator on $\Sigma \times \Sigma$ defined as
	\begin{alignat*}{2}
		\Gamma(X) =~ &\{(\sigma,\sigma') ~|~ \calM,\sigma \models B ~\text{and}~ (\sigma,\sigma') \in X \circ C(p_0)\} ~\cup \\
		&\{(\sigma,\sigma) ~|~ \calM,\sigma \models \neg B\}.
	\end{alignat*}

We see that $\Gamma$ is a monotone operator and thus $\lfp(\Gamma)$ is well-defined.
Moreover we can convince ourselves that $C(p)$ is actually a partial function from $\Sigma \rightarrow \Sigma$. If $C(p)(\sigma)$ is not defined, it means that a while-loop is not terminating. To make it a total function, we extend the set of states $\Sigma$ with the state $\bot$, which is associated with a non-terminating computation, i.e. we define $\Sigma_{\bot} := \Sigma \cup \{\bot\}$. For every $\sigma$ such that $C(p)(\sigma)$ is not defined yet, we define $C(p)(\sigma) := \bot$ and in that way $C(p)$ is a total function from $\Sigma \rightarrow \Sigma_{\bot}$.

\end{definition}

Now we define the meaning of partial correctness of a Hoare triple:

\begin{definition}[Semantics of Hoare triples]
	Let $\calL$ be a language and $\calM$ be a structure. Let $\sigma$ be a state, for a Hoare triple $\hoare{\varphi}{p}{\psi}$ define 
	\begin{align*}
		\calM, \sigma \models \hoare{\varphi}{p}{\psi} \quad \text{if} \quad  \calM,\sigma \models \varphi ~\Rightarrow~ \calM,C(p)(\sigma) \models \psi.
	\end{align*}
	Here $\calM,\bot \models \psi$ is defined to be true. In doing so we only check validity for states, in which the computation terminates. Now define
	\begin{align*}
		\calM \models \hoare{\varphi}{p}{\psi} \quad \text{if} \quad
		\forall \sigma \in \Sigma:\ \calM,\sigma \models \hoare{\varphi}{p}{\psi}.
	\end{align*}

Note that we can define $\calM \models \hoare{R}{p}{S}$ analogously for relations $R$ and $S$. This will be of use in some proofs.

\end{definition}

\subsection{Verification condition}

We define a different semantics of Hoare triples using the verification condition, which turns out to be equivalent to the usual semantics. As the verification condition is a linear-Horn formula equation, we are able to apply our results from Section \ref{sec.afpthm}. In particular we are able to express partial correctness of a Hoare triple as an FO[LFP]-formula.

\begin{definition}
	The \emph{verification condition} of a Hoare triple $\hoare{\varphi}{p}{\psi}$, written $\vc(\hoare{\varphi}{p}{\psi})$, is a formula equation $\exists \mybar{I} \forall \mybar{x}~ \vct(\hoare{\varphi}{p}{\psi})$, where $\vct(\hoare{\varphi}{p}{\psi})$ is defined by structural induction on $p$. Here $I$ is a fresh new formula variable, which does not appear in $\varphi$ nor $\psi$.
	\begin{align*}
		\vct(\hoare{\varphi}{\hskipit}{\psi}) &= (\varphi \rightarrow \psi)\\
		\vct(\hoare{\varphi}{x_j := t}{\psi}) &= (\varphi \rightarrow \psi[x_j\backslash t])\\
		\vct(\hoare{\varphi}{p_0;p_1}{\psi}) &= \vct(\hoare{\varphi}{p_0}{I}) \wedge \vct(\hoare{I}{p_1}{\psi}), \\
		\vct(\hoare{\varphi}{\hif B \hthen p_0 \helse p_1}{\psi}) &= \vct(\hoare{\varphi\wedge B}{p_0}{\psi}) \wedge \vct(\hoare{\varphi\wedge \neg B}{p_1}{\psi})\\
		\vct(\hoare{\varphi}{\hwhile B \hdo p_0}{\psi}) &= \vct(\hoare{I \wedge B}{p_0}{I}) \wedge (\varphi \rightarrow I) \wedge (I \wedge \neg B \rightarrow \psi)
	\end{align*}
	Then $\vc(\hoare{\varphi}{p}{\psi}) = \exists \mybar{I} \forall \mybar{x}~ \vct(\hoare{\varphi}{p}{\psi})$ is defined by universal quantification of every individual variable occurring in $\vct(\hoare{\varphi}{p}{\psi})$ and existential quantification of every predicate variable in $\vct(\hoare{\varphi}{p}{\psi})$.
\end{definition}
Note that there are translations which are more efficient in practice, for example in
 dealing with composition.
However, since we aim for theoretical results in this paper, we have opted for this simple
 translation.
Also note that this is a purely syntactic definition, thus we can define
 $\vc(\hoare{R}{p}{S})$ analogously for a program $p$ and predicate variables $R,S$.

\begin{lemma}
	Let $\hoare{\varphi}{p}{\psi}$ be Hoare triple. Then $\vc(\hoare{\varphi}{p}{\psi})$ is a linear-Horn formula equation.
\end{lemma}
\begin{proof}
	The verification condition is a formula equation $\exists \mybar{I} \forall \mybar{x}~ \vct(\hoare{\varphi}{p}{\psi})$, where $\vct(\hoare{\varphi}{p}{\psi})$ is a conjunction of clauses. Each clause has the form $\gamma \vee \neg C \vee D$, where $\gamma$ is a first-order formula (it is a disjunction of the first-order formulas $B$ or $\neg B$) and $C$ and $D$ are either predicate variables or the first-order formulas $\varphi$ or $\psi$.
\end{proof}

\begin{lemma}%\label{HoareVcPositively}
	Let $\hoare{\varphi}{p}{\psi}$ be a Hoare triple. Then $\varphi$ occurs only negatively and $\psi$ occurs only positively in $\vc(\hoare{\varphi}{p}{\psi})$, respectively.
\end{lemma}

\begin{theorem}\label{thm.vcSoundness}
	Let $\calL$ be a language and $\calM$ be a structure. Let $\hoare{\varphi}{p}{\psi}$ be a Hoare triple. Then 
	\[
	\calM \models \vc(\hoare{\varphi}{p}{\psi}) \quad \Rightarrow \quad  \calM \models \hoare{\varphi}{p}{\psi}.
	\]
\end{theorem}
\begin{proof}
This theorem is shown by structural induction on $p$. For a detailed proof see \cite{Kloibhofer20Fixed}.
\end{proof}

The next aim is to show the converse direction of Theorem \ref{thm.vcSoundness}, i.e. that $\calM \models \hoare{\varphi}{p}{\psi}$ implies  $\calM \models \vc(\hoare{\varphi}{p}{\psi})$. To prove this we need the concept of the weakest precondition.

\begin{definition}\label{def.weakPre}
	Let $\calL$ be a language and $\calM$ be a structure. Let $p$ be a program and $\varphi, \psi$ be first-order formulas in $\calL$. The \emph{weakest precondition}\footnote{In the literature this is mostly called weakest liberal precondition and the term weakest precondition is reserved for the context of total correctness. As we only talk about relative correctness of programs there is no need for us to do so.} of $p$ and $\psi$, written $\wp(p,\psi)$, is defined as
	\begin{align*}
		\wp(p,\psi) = \{\sigma \in \Sigma ~|~ \calM,C(p)(\sigma) \models \psi \}.
	\end{align*}
	For technical reasons we also define the relation $\Rwp(p,\psi)$ defined by the weakest precondition, i.e.
	\[
	\calM,\sigma \models \Rwp(p,\psi) ~\Leftrightarrow~ \calM,C(p)(\sigma) \models \psi.
	\]

	The \emph{strongest postcondition} of $p$ and $\varphi$, written $\sp(p,\varphi)$, is defined as
	\begin{align*}
		\sp(p,\varphi) = \{\sigma \in \Sigma ~|~ \exists \sigma' \in \Sigma: \calM,\sigma' \models \varphi ~\text{and}~ C(p)(\sigma') = \sigma \}.
	\end{align*}
\end{definition}

\begin{definition}
	Let $\calL$ be a language and $\calM$ be a structure. For every formula $\varphi$ in $\calL$ define a subset of $\Sigma$:
	\begin{align*}
		[\varphi] := \{\sigma \in \Sigma ~|~ \calM,\sigma \models \varphi \}.
	\end{align*}
	Similarly we define for a relation $R$ in $\calM$
		\begin{align*}
		[R] := \{\sigma \in \Sigma ~|~ \calM,\sigma \models R \}.
	\end{align*}
	In particular $[\Rwp(p,\psi)] = \wp(p,\psi)$.
\end{definition}

The following properties of the weakest precondition and the strongest postcondition
justify the terminology and are of fundamental importance.

\begin{lemma}\label{lem.weakPreEquiv} 
	Let $\calL$ be a language and $\calM$ be a structure. Let $\hoare{\varphi}{p}{\psi}$ be a Hoare triple. Then
	\begin{enumerate}
		\item $\calM \models \hoare{\varphi}{p}{\psi}$ iff $[\varphi] \subseteq \wp(p,\psi)$,
		\item $\calM \models \hoare{\varphi}{p}{\psi}$ iff $ [\psi] \supseteq \sp(p,\varphi)$.
	\end{enumerate}
\end{lemma}

\begin{lemma}\label{lem.wpHoareTripel}
	Let $\calL$ be a language and $\calM$ be a structure. Let $p$ be a program and $\psi$ be a first-order formula. Then $\calM \models \vc(\hoare{R_{\wp}(p,\psi)}{p}{\psi})$
\end{lemma}
\begin{proof}
	This is proven by structural induction on the program $p$. We prove the cases $p \equiv p_0;p_1$ and $p \equiv \hwhile B \hdo p_0$, the others are routine.

	 Let $p \equiv p_0;p_1$. We ought to show that $\calM \models \exists I ~ \vc(\hoare{\Rwp(p,\psi)}{p_0}{I}) \wedge \vc(\hoare{I}{p_1}{\psi})$. By the induction hypothesis it holds
	 \begin{align*}
	 	\calM \models& \vc(\hoare{\Rwp(p_1,\psi)}{p_1}{\psi}) \quad \text{and}\\
	 	\calM \models& \vc(\hoare{\Rwp(p_0,\Rwp(p_1,\psi))}{p_0}{\Rwp(p_1,\psi)}).
	 \end{align*}
 	We are finished if we can show that $\Rwp(p,\psi) = \Rwp(p_0,\Rwp(p_1,\psi))$, yet this follows from
 	\begin{align*}
 		\calM, \sigma \models \Rwp(p,\psi) \quad 
 		\Leftrightarrow \quad&  \calM, C(p_0;p_1)(\sigma) \models \psi \\
 		\Leftrightarrow \quad& \calM, C(p_1) \circ C(p_0)(\sigma) \models \psi \\
 		\Leftrightarrow \quad& \calM, C(p_0)(\sigma) \models \Rwp(p_1,\psi) \\
 		\Leftrightarrow \quad& \calM, \sigma \models \Rwp(p_0,\Rwp(p_1,\psi)).
 	\end{align*}

	Now consider the case $p \equiv \hwhile B \hdo p_0$. For shorter notation we define $R = \Rwp(p,\psi)$. We have to show that $\calM \models \exists I ~\vc(\hoare{I\wedge B}{p_0}{I}) \wedge (R \impl I) \wedge (I\wedge \neg B \impl \psi)$. It suffices to show that
	\begin{enumerate}
		\item $\calM \models \vc(\hoare{R \wedge B}{p_0}{R})$ and
		\item $\calM \models (R \wedge \neg B) \impl  \psi$.
	\end{enumerate}

To prove these two claims we use the following fact, cf. \cite[Proposition 5.1.]{Winskel93Formal}, for a program of the form $p \equiv \hwhile B \hdo p_0$:
\begin{equation}\label{propWinskelWhile}
	C(p) \equiv C(\hif B \hthen p_0 ; p \helse \hskipit)
\end{equation}

We now prove (1). The induction hypothesis states $\calM \models \vc(\hoare{\Rwp(p_0,R)}{p_0}{R})$, hence it remains to show that $\calM \models (R \wedge B) \impl \Rwp(p_0,R)$. This follows as 

 \begin{align*}
	\calM, \sigma \models \Rwp(p,\psi) \wedge B \quad 
	\Leftrightarrow \quad&  \calM, C(p)(\sigma) \models \psi \wedge \calM,\sigma \models B \\
	\overset{(\ref{propWinskelWhile})}{\Rightarrow} \quad& \calM,  C(p_0,p)(\sigma) \models \psi \\
	\Leftrightarrow \quad& \calM, C(p) \circ C(p_0)(\sigma) \models \psi \\
	\Leftrightarrow \quad& \calM, C(p_0)(\sigma) \models \Rwp(p,\psi) \\
	\Leftrightarrow \quad& \calM, \sigma \models \Rwp(p_0,\Rwp(p,\psi)).
\end{align*}

For the second claim let $\sigma$ be a state such that $\calM,\sigma \models R \wedge \neg B$. By the definition of $R$ we have $\calM,C(p)(\sigma) \models \psi$, thus (\ref{propWinskelWhile}) yields $\calM, \hskipit(\sigma) \models \psi$, which translates to $\calM, \sigma \models \psi$.

\end{proof}

\begin{theorem}\label{thm.vcCompleteness}
	Let $\calL$ be a language and $\calM$ be a structure. Let $\hoare{\varphi}{p}{\psi}$ be a Hoare triple. Then 
	\[
	\calM \models \hoare{\varphi}{p}{\psi} \quad \Rightarrow \quad \calM \models \vc(\hoare{\varphi}{p}{\psi}).
	\]
\end{theorem}
\begin{proof}
	By Lemma \ref{lem.weakPreEquiv} the assumption $\calM \models \hoare{\varphi}{p}{\psi}$ is equivalent to $\calM \models \varphi \impl R_{\wp}(p,\psi)$. Lemma \ref{lem.wpHoareTripel} states that $\calM \models \vc(\hoare{R_{\wp}(p,\psi)}{p}{\psi})$ and, as $R_{\wp}(p,\psi)$ occurs only negatively in $\vc(\hoare{R_{\wp}(p,\psi)}{p}{\psi})$, Lemma \ref{lem.ImplPositiv} concludes $\calM \models \vc(\hoare{\varphi}{p}{\psi})$.
\end{proof}

As $\vc(\hoare{\varphi}{p}{\psi})$ is a linear-Horn formula equation we can now obtain the statement, that a partial correctness assertion of an imperative program
is expressible as a formula in FO[LFP], a point also made in~\cite{Blass87Existential} for existential least fixed-point logic.
\begin{corollary}\label{HoareVCLinearCorollary}
	Let $\calL$ be a language and $\calM$ be a structure. Let $\hoare{\varphi}{p}{\psi}$ be a Hoare triple. Consider the linear-Horn formula equation $\exists \mybar{I} ~\pi \equiv \vc(\hoare{\varphi}{p}{\psi})$ with predicate variables $I_1,...,I_n$. Define $\mu_j := [\lfp_{I_j} ~\Phi_{\pi}]$ and $\nu_j := \neg [\lfp_{I_j} ~\Phi_{\pi^D}]$ for $j \in \{1,...,n\}$. Then 
	\begin{enumerate}
		\item $\calM \models \vc(\hoare{\varphi}{p}{\psi})$ iff $\calM \models \vct(\hoare{\varphi}{p}{\psi})[\mybar{I}\backslash \mybar{\mu}]$ iff $\calM \models \vct(\hoare{\varphi}{p}{\psi})[\mybar{I}\backslash \mybar{\nu}]$.
		\item If $\calM \models \vct(\hoare{\varphi}{p}{\psi})[\mybar{I}\backslash \mybar{R}]$ for relations $R_1,...,R_n$, then $\calM \models \bigwedge_{j=1}^n \mu_j \rightarrow R_j \wedge R_j \rightarrow \nu_j$. 
	\end{enumerate}
\end{corollary}

\subsection{Weakest Precondition and Strongest Postcondition}

By defining the equivalent semantics based on the verification condition, which is a linear-Horn formula equation, we now get some corollaries of our linear-Horn fixed-point theorem. We will see that the canonical solutions correspond to the weakest precondition and strongest postcondition. 

In the integers it is well-known that for any program $p$ and any formula $\psi$ there is a first-order formula
$\varphi_\wpc$ which defines the weakest precondition, i.e., $[\varphi_\wpc] = \wpc(p,\psi)$ and, symmetrically,
for any program $p$ and any formula $\varphi$ there is a first-order formula $\psi_\spc$ which
defines the strongest postcondition, i.e. $[\psi_\spc] = \spc(p,\varphi)$.
Note that these formulas rely on the expressivity of the assertion language, i.e, in this setting, 
on an encoding of finite sequences in $\Int$. By allowing FO[LFP] formulas we are not dependent on
 the expressivity of the structure and moreover the FO[LFP] formulas more clearly describe the runs of the program. 
This point was made prominently in~\cite{Blass87Existential}.

Consider the formula $\exists X ~\vc(\hoare{\varphi}{p}{X})$, which asks for a relation $X$ such that all states satisfying $\varphi$ fulfil $X$ after running the program $p$. Similarly we are interested in solutions of the formula $\exists Y ~\vc(\hoare{Y}{p}{\psi})$. Note that these are linear-Horn formula equations and therefore we can apply the results from Section \ref{sec.afpthm}. In general there also occur predicate variables in $\vc(\hoare{\varphi}{p}{\psi})$, yet here we are only interested in the predicate variables that are stated specifically.

\begin{theorem}\label{thm.wpLfpNu}
	Let $\calL$ be a language and $\calM$ be a structure. Let $p$ be a program and $\psi$ be a first-order formula. Consider the linear-Horn formula equation $\exists Y ~\pi \equiv \exists Y~ \vc(\hoare{Y}{p}{\psi})$. Let $\nu := \neg [\lfp_{Y} ~\Phi_{\pi^D}]$, then 
	\begin{align*}
		[\nu] = \wp(p,\psi).
	\end{align*}
\end{theorem}
\begin{proof}
	As $\calM \models \vc(\hoare{\bot}{p}{\psi})$ we know that $\calM \models \exists Y~ \vc(\hoare{Y}{p}{\psi})$.  Hence with Theorem \ref{thm.abstractDualFP}/1 we obtain $\calM \models \vc(\hoare{\nu}{p}{\psi})$, which implies $\calM \models \hoare{\nu}{p}{\psi}$ due to Theorem \ref{thm.vcSoundness}. Now Lemma \ref{lem.weakPreEquiv} states $[\nu] \subseteq \wp(p,\psi)$.
	
	For the other inclusion let $\Rwp(p,\psi)$ be the relation defined by the weakest precondition, it holds $[\Rwp(p,\psi)] = \wp(p,\psi)$ . From Lemma \ref{lem.wpHoareTripel} we obtain $\calM \models \vc(\hoare{\Rwp(p,\psi)}{p}{\psi})$ and we can apply Theorem \ref{thm.abstractDualFP}/2. This yields $\calM \models \forall \mybar{x}~ (\Rwp(p,\psi)(\mybar{x}) \rightarrow \nu(\mybar{x}))$. In particular we got $\calM, \sigma \models \Rwp(p,\psi) \rightarrow \nu$ for all $\sigma \in \Sigma$. Hence $[\nu] \supseteq [\Rwp(p,\psi)] = \wp(p,\psi)$.
\end{proof}

\begin{theorem}\label{thm.spLfpMu}
	Let $\calL$ be a language and $\calM$ be a structure. Let $p$ be a program and $\varphi$ be a first-order formula. Consider the linear-Horn formula equation $\exists X ~\pi \equiv \exists X~ \vc(\hoare{\varphi}{p}{X})$. Let $\mu := [\lfp_{X} ~\Phi_{\pi}]$, then 
	\begin{align*}
		[\mu] = \sp(p,\varphi).
	\end{align*}
\end{theorem}
\begin{proof}
	Dual to the proof of Theorem \ref{thm.wpLfpNu}
\end{proof}

Note that the formulas $\mu$ and $\nu$ thus obtained do not rely on an expressivity hypothesis anymore.
The encoding of sequences is replaced by the least fixed-point operator.

\subsection{Hoare logic}\label{subsect.HoareLogic}

It is well known that Hoare logic, as introduced in \cite{Hoare1969}, is a sound and relative complete proof system for Hoare triples. Hoare logic is usually formulated in the language of arithmetic $\calL_A$ and for the structure $\Int$. We generalise this definition for an arbitrary structure and show its equivalence to the abstract semantics of the verification condition of a Hoare triple. In the classical setting of $\Int$ we obtain the usual soundness and completeness result, where the use of the expressivity hypothesis is pointed out explicitly.

Since many practical algorithms for verification generate invariants of a particular form, we introduce
 the Hoare calculus parameterised by a set of first-order formulas $\calC$ which denotes the formulas
 allowed as invariants.
In Theorem~\ref{thm.hoareVcEquiv} below we show that provability in the Hoare calculus restricted to $\calC$
 is equivalent to the truth of the verification condition in the definable model abstraction of $\calC$.

\begin{comment}

Let $\calL$ be a language and $\calM$ be a model. In general a set of first-order formulas $\calC$ does not define a model abstraction of $\calM$, as the definable sets by $\calC$ do not have to form a lattice. Yet if we are only interested in the interpretation of second-order quantifiers and disregard least fixed-point atoms, we can state a similar definition. This leads to the notion of second-order model abstractions, which is reminiscent of Henkin semantics \cite{Henkin50Completeness}. 

\begin{definition}
	Let $\calL$ be a language and $\calM$ be a structure. A \emph{second-order model abstraction} is a pair $(\calM,\calC)$, where $\calC$ is a set of first-order formulas in $\calL$. The semantics of second-order model abstraction is defined similarly as for model abstractions. First-order formulas are interpreted in $\calM$ as usual. For formulas of the form $\exists X\, \psi$, where $X$ is a predicate variable, we define
	\[
	(\calM,{\calC}),\theta \modelsa \exists X \psi \quad \Leftrightarrow \quad \exists \chi \in \calC: (\calM,\calC),\theta \modelsa \psi\unsubst{X}{\chi},
	\]
	and analogously for formulas of the form $\forall X \psi$.
\end{definition}

\end{comment}

\begin{definition}[Hoare rules]
	Let $\calL$ be a language, $\calM$ be a model and $\calC$ be a set of first-order formulas in $\calL$. Let $\hoare{\varphi}{p}{\psi}$ be a Hoare triple. We define the following rules, called Hoare rules. 
	We write $\calM \vdash_{\calC}\hoare{\varphi}{p}{\psi}$ if $\hoare{\varphi}{p}{\psi}$ is provable by the Hoare rules in $\calM$ and $\calC$. If $\calC$ is the set of all first-order formulas we often omit the subscript and write $\calM \vdash\hoare{\varphi}{p}{\psi}$. Here $B$ is a first-order formula in $\calL$, $p,p_0$ and $p_1$ are programs, $t$ is an $\calL$-term, $x_j$ is a variable and, most importantly, $I$ is in $\calC$.
	
	\begin{center}
		\begin{prooftree}
			\infer0[\text{\emph{(skip)}}]{\hoare{\varphi}{\hskipit}{\varphi}}
		\end{prooftree}
	\end{center}
	\smallskip
	\begin{center}
		\begin{prooftree}
			\infer0[\text{\emph{(assign)}}]{\hoare{\varphi[x_j \backslash t]}{x_j:=t}{\varphi}}
		\end{prooftree}
	\end{center}
	\smallskip
	\begin{center}
		\begin{prooftree}
			\hypo{\hoare{\varphi}{p_0}{I}}	
			\hypo{\hoare{I}{p_1}{\psi}}
			\infer2[\text{\emph{(sequence)}}]{\hoare{\varphi}{p_0;p_1}{\psi}}
		\end{prooftree}
	\end{center}
	\smallskip
	\begin{center}
		\begin{prooftree}
			\hypo{\hoare{\varphi \wedge B}{p_0}{\psi}}	
			\hypo{\hoare{\varphi \wedge \neg B}{p_1}{\psi}}
			\infer2[\text{\emph{(conditional)}}]{\hoare{\varphi}{\hif B \hthen p_0 \helse p_1}{\psi}}
		\end{prooftree}
	\end{center}
	\smallskip
	\begin{center}
		\begin{prooftree}
			\hypo{\calM \models \varphi \rightarrow I}
			\hypo{\hoare{I \wedge B}{p}{I}}	
			\hypo{\calM \models I \wedge \neg B \rightarrow \psi}
			\infer3[\text{\emph{(while)}}]{\hoare{\varphi}{\hwhile B \hdo p}{\psi}}
		\end{prooftree}
	\end{center}
	\smallskip
	\begin{center}
		\begin{prooftree}
			\hypo{\calM \models \varphi \rightarrow \varphi'}
			\hypo{\hoare{\varphi'}{p}{\psi'}}	
			\hypo{\calM \models \psi' \rightarrow \psi}
			\infer3[\text{\emph{(consequence)}}]{\hoare{\varphi}{p}{\psi}}
		\end{prooftree}
	\end{center}
	
\end{definition}

\begin{remark}
In most of the literature the (while)-rule is replaced by
\begin{center}
	\begin{prooftree}
		\hypo{\hoare{\varphi \wedge B}{p}{\varphi}}	
		\infer1[\text{\emph{(while')}}]{\hoare{\varphi}{\hwhile B \hdo p}{\varphi \wedge \neg B}}.
	\end{prooftree}
\end{center}
If $\calC$ is the set of all first-order formulas this leads to an equivalent proof calculus, as (while') is a special case of (while)  for $I = \varphi$ and $\psi = \varphi \wedge \neg B$. In the other direction (while) can be obtained by combining the (while') and (consequence) rule. Observe that this is not true in general as (while') is not a special case of (while) if $\varphi \notin \calC$.
\end{remark}

In the next theorem we relate Hoare logic with the validity of the verification condition of a Hoare triple. The relation will be fine-grained, depending on a set of first-order formulas $\calC$. We first define the semantic relation $\calM \models_{\calC} \psi$, where the second-order quantifiers in $\psi$ are interpreted over the formulas in $\calC$.

Note that, if $(\calM,G_{\calC})$ is a definable model abstraction, then $(\calM,G_{\calC}) \modelsa \psi$ iff $\calM \models_{\calC} \psi$ for all SO-formulas $\psi$. The difference is that for the relation $\models_{\calC}$ we do not demand that the sets defined by formulas in $\calC$ form a complete lattice, and thus fixed-point operators can not be interpreted.
\begin{definition}
	Let $\calC$ be a set of first-order formulas. Given a model $\calM$ and environment $\theta$, we define the relation $\calM, \theta \models_{\calC} \phi$ inductively on SO-formulas $\phi$. If $\phi$ is first-order, then $\models_{\calC}$ is identical to $\models$. For formulas of the form $\exists X \psi$, we define 
	\[
		\calM, \theta \models_{\calC} \exists \psi \quad \Leftrightarrow \quad \exists \chi \in \calC:~ \calM,\theta[X := \chi] \models_{\calC} \psi,
	\]
	and analogously for formulas of the form $\forall X \psi$. As usual, we define $\calM \models_{\calC} \psi$ if $\calM,\theta \models_{\calC} \psi$ for all environments $\theta$.
\end{definition}

\begin{theorem}\label{thm.hoareVcEquiv}
	Let $\calL$ be a language and $\calM$ be a model. Let $\hoare{\varphi}{p}{\psi}$ be a Hoare triple and let $\calC$ be a set of first-order formulas. % in $\calL$ such that\ $(\calM,G_\calC)$ is a definable model abstraction. 
	Then
	\[
		\calM \models_{\calC} \vc(\hoare{\varphi}{p}{\psi})\quad \Leftrightarrow \quad \calM \vdash_{\calC} \hoare{\varphi}{p}{\psi}.
	\]
\end{theorem}
\begin{proof}
	$``\Rightarrow":$ The proof is by induction on the program $p$. We only show the cases $p \equiv \hskipit$ and $p \equiv \hwhile B \hdo p_0$, the others are similar.
	\begin{itemize}
		\item $p \equiv \hskipit$: The assumption is $\calM \models \varphi \impl \psi$ and the (skip)-rule yields $\calM \vdash \hoare{\varphi}{\hskipit}{\varphi}$. Hence we obtain $\calM \vdash \hoare{\varphi}{\hskipit}{\psi}$ with the (consequence)-rule. Note that in the (consequence)-rule $\varphi'$ and $\psi'$ may be any first-order formulas and are not restricted to $\calC$, which is of importance here.
		\item $p \equiv \hwhile B \hdo p_0$: The assumption is $\calM \models_{\calC} \exists I ~\vc(\hoare{I \wedge B}{p_0}{I}) \wedge (\varphi \rightarrow I) \wedge (I \wedge \neg B \rightarrow \psi)$, which means that there exists $\chi \in \calC$ such that $\calM \models_{\calC} \vc(\hoare{\chi \wedge B}{p_0}{\chi})$ and $\calM \models (\varphi \rightarrow \chi) \wedge (\chi \wedge \neg B \rightarrow \psi)$. By the induction hypothesis we have $\calM \vdash_{\calC} \hoare{\chi \wedge B}{p_0}{\chi}$, now we can use the (while)-rule to obtain $\calM \vdash_{\calC} \hoare{\varphi}{p}{\psi}$.
	\end{itemize}

	$``\Leftarrow":$ This direction is shown by proving inductively, that every Hoare rule yields Hoare triples such that $\calM \models_{\calC} \vc(\hoare{\varphi}{p}{\psi})$. In other words, we show that the Hoare rules are sound with respect to the semantics $\calM \models_{\calC} \vc(\hoare{\varphi}{p}{\psi})$. We concentrate on the exemplary (sequence) and (consequence)-rule.
	\begin{itemize}
		\item (sequence): By assumption there exists $I \in \calC$ such that $\calM \vdash_{\calC} \hoare{\varphi}{p_0}{I}$ and $\calM \vdash_{\calC} \hoare{I}{p_1}{\psi}$. Using the induction hypothesis we obtain that there exists $I \in \calC$ such that $\calM \models_{\calC} \vc(\hoare{\varphi}{p_0}{I}) \wedge \vc(\hoare{I}{p_1}{\psi})$, which is the definition of $\calM \models_{\calC} \vc(\hoare{\varphi}{p_0;p_1}{\psi})$.
		\item (consequence): By assumption $\calM \models \varphi \impl \varphi'$, $\calM \models \psi' \impl \psi$ and $\calM \proves_{\calC} \vc(\hoare{\varphi'}{p}{\psi'})$. Using the induction hypothesis we have $\calM \models_{\calC} \vc(\hoare{\varphi'}{p}{\psi'})$. Now Lemma \ref{lem.ImplPositiv} yields $\calM \models_{\calC} \vc(\hoare{\varphi}{p}{\psi})$. Note that we stated Lemma \ref{lem.ImplPositiv} only for classical semantics, yet it straightforwardly generalizes to the relation $\models_{\calC}$.
	\end{itemize}
\end{proof}

\begin{lemma}\label{lem.abstractVc}
	Let $\calL$ be a language and $\calM$ be a model. Let $\hoare{\varphi}{p}{\psi}$ be a Hoare triple and let $\calC$ be a set of first-order formulas in $\calL$. Then
	\[
	\calM \models_{\calC} \vc(\hoare{\varphi}{p}{\psi}) \quad \Rightarrow \quad \calM \models \vc(\hoare{\varphi}{p}{\psi}).
	\]
\end{lemma}
\begin{proof}
	This follows as $\vc(\hoare{\varphi}{p}{\psi})$ is an existential second-order formula.
\end{proof}

Notably in the traditional case of the integers and the set of all first-order formulas $\calC$ the converse of Lemma \ref{lem.abstractVc} holds as well. Here the expressivity of $\Int$ is crucial.

\begin{lemma}\label{lem.abstractInt}
	Let $\calL_A = \{0,1,+,-,\cdot,\leq\}$ be the language of arithmetic and $\calC$ be the set of all first-order formulas. Then 
	\[
		\Int \models \vc(\hoare{\varphi}{p}{\psi}) \quad \Rightarrow \quad \Int \models_{\calC} \vc(\hoare{\varphi}{p}{\psi}).
	\]
\end{lemma}
\begin{proof}
	The fixed-point theorem Corollary \ref{cor.HornFP} states, that $\Int \models \vc(\hoare{\varphi}{p}{\psi})$ is equivalent to $\Int \models \vct(\hoare{\varphi}{p}{\psi})\unsubst{\mybar{I}}{\mybar{\mu}}$, where $\mybar{\mu}$ is a tuple of LFP-atoms. Using Gödel's $\beta$-function we can encode those LFP-atoms and find first-order formulas $\lambda_1,...,\lambda_n$ such that $\Int \models \mu_j \leftrightarrow \lambda_j$ for $j = 1,...,n$. We explain how this can be done for an LFP-atom $\mu \equiv [\lfp_R ~ \phi]$, where $\phi$ is an existential first-order formula containing $R$. For $l \in \omega$ let $\sigma_\phi^l$ be defined as in Definition \ref{def.apprSigma}. Theorem \ref{thm.apprLfp} states that 
	\[[\lfp_R ~ \phi] \equiv \bigvee_{l\in \omega} \sigma_\phi^l. \]
	Now define $\lambda(\mybar{x})$ using Gödel's $\beta$-function as follows: ``There exists a sequence of formulas $\sigma_\phi^0,...,\sigma_\phi^n$ defined as above such that $\sigma_\phi^n(\mybar{x})$ holds''. Then $\Int \models \mu \liff \lambda$. 
	As $\mybar{\mu}$ is defined from existential first-order formulas this applies here. 	
	Hence $\Int \models \vct(\hoare{\varphi}{p}{\psi}) \unsubst{\mybar{I}}{\mybar{\lambda}}$, which implies $\Int \models_{\calC} \vc(\hoare{\varphi}{p}{\psi})$.
\end{proof}

\begin{remark}\label{rem.hoareVcEquiv}
	 Theorem \ref{thm.hoareVcEquiv} together with Lemma \ref{lem.abstractVc} and \ref{lem.abstractInt} yields
	\[
	\Int \models \vc(\hoare{\varphi}{p}{\psi})\quad \Leftrightarrow \quad \Int \vdash \hoare{\varphi}{p}{\psi}.
	\]
	This has also been shown in \cite{Kloibhofer20Fixed}, yet here we pointed out explicitly where the expressivity of $\Int$ is used and which parts of the proof can be generalised to different structures. Combining this equality with Theorem \ref{thm.vcSoundness} and \ref{thm.vcCompleteness} we obtain the classical statement
	\[
	\Int \models \hoare{\varphi}{p}{\psi}\quad \Leftrightarrow \quad \Int \vdash \hoare{\varphi}{p}{\psi}.
	\]
\end{remark}

\section{Fixed-point approximation}\label{sec.approximation}

The problem of finding first-order formulas which approximate a second-order formula is an
 intensively studied topic in the history of logic.
For second-order formulas of the form $\exists \mybar{X} \forall \mybar{y}~ \psi$, where $\psi$ is
 quantifier-free, it has already been investigated by Ackermann in 1935~\cite{Ackermann35Untersuchungen}.
Under the assumptions that the language $\calL$ is relational, i.e., $\calL$ contains no
 function symbols, and there is only one unary predicate variable $X$,
 \cite{Ackermann35Untersuchungen} shows in detail how to compute
 an infinite conjunction of first-order formulas that is equivalent to the second-order formula
 $\exists X \forall \mybar{y}~ \psi$.
This is achieved with a method similar to modern resolution.
This result has also been extended to any number of predicate variables of arbitrary
 arity~\cite{Ackermann35Untersuchungen,Wernhard17Approximating}, yet the assumption of
 a relational language remains.
In this section we show how to obtain a similar result for any language $\calL$, but with another assumption:
 we only consider Horn formula equations.
Moreover, this is attained with a completely different method as a straightforward corollary of our
 Horn fixed-point theorem.

Let $\exists \mybar{X} \psi$ be a Horn formula equation. In the last section we found least fixed-point logic formulas $\mu_1,...,\mu_n$ such that $\models \exists \mybar{X}  ~\psi \leftrightarrow \psi[X_1\backslash \mu_1,...,X_n\backslash\mu_n]$. We can now use the first-order formulas defined in Section \ref{sec.prelim.approximation}, which approximate $\mu_1,...,\mu_n$ and therefore lead to an approximation of the second-order formula $\exists \mybar{X} \psi$. 

\begin{theorem}\label{thm.approximation}
	Let $\exists \mybar{X} \psi$ be a Horn formula equation. Then there exists a, possibly infinite, set of first-order formulas $\Psi$ such that 
	\begin{align*}
		\exists \mybar{X} \psi \equiv \bigwedge_{\varphi \in \Psi} \varphi.
	\end{align*}	
\end{theorem}
\begin{proof}
	Let $\mu_j := [\lfp_{X_j} ~\Phi_{\psi}]$
	for $j \in \{1,\ldots,n\}$.	Applying Corollary \ref{cor.HornFP} we have 
	\[
		\exists \mybar{X}\,\psi \equiv \psi[X_1 \backslash \mu_1,...,X_n \backslash \mu_n]
	\]
	By construction $\mybar{\mu}$ satisfies all base and induction clauses, thus
	\begin{align*}
		\exists \mybar{X} \psi \equiv \forall \mybar{y} \bigwedge_{C \in E} C[X_1 \backslash \mu_1,...,X_n \backslash \mu_n](\mybar{y}),
	\end{align*}
	where $E$ is the set of all end clauses in $\psi$. Note that the formulas $\varphi_1,...,\varphi_n$ in the $n$-tuple $\Phi$ are existential first-order formulas, hence we can use Theorem \ref{thm.apprLfp} to obtain
	\[
	\exists \mybar{X} \psi \equiv \forall \mybar{y} \bigwedge_{C \in E} C[X_1 \backslash \sigma_{1,\Phi}^{\omega},...,X_n \backslash \sigma_{n,\Phi}^{\omega}](\mybar{y}),
	\]
	where $\sigma_{j,\Phi}^{\omega} \equiv \bigvee_{l\in \omega} \sigma_{j,\Phi}^l$ is an infinite disjunction of first-order formulas for $j \in \{1,...,n\}$. The clauses in $E$ have the form $\neg \gamma \vee \neg X_{i_1}(\mybar{t_1}) \vee \dots \vee \neg X_{i_m}(\mybar{t_m})$. We denote a clause in $E$ by the determining tuple $(\gamma,\iota,\tau)$, where $\iota := \{i_1,...,i_m\}$ and $\tau := \{\mybar{t_1},...,\mybar{t_m}\}$. Then 
	\begin{align*}
		\exists \mybar{X} \psi \equiv \forall \mybar{y} \bigwedge_{(\gamma,\iota,\tau) \in E} \left(\neg \gamma \vee \neg \sigma_{i_1,\Phi}^{\omega}(\mybar{t_1}) \vee \dots \vee \neg \sigma_{i_m,\Phi}^{\omega}(\mybar{t_m})\right).
	\end{align*}
	Per definition there is a set of first-order formulas $\Psi_j$ such that $\sigma_{j,\Phi}^{\omega} \equiv \bigvee_{\varphi \in \Psi_j} \varphi$ for every $j \in \{1,...,n\}$. Thus with repeated application of Lemma \ref{lem.infiniteFormulas} there exists a set of first-order formulas $\Psi$ such that 
	\begin{align*}
		\forall \mybar{y} \bigwedge_{(\gamma,\iota,\tau) \in E} \left(\neg \gamma \vee \neg (\bigvee_{\varphi \in \Psi_{i_1}} \varphi(\mybar{t_1})) \vee \dots \vee \neg (\bigvee_{\varphi \in \Psi_{i_m}} \varphi(\mybar{t_m}))\right) \equiv \bigwedge_{\varphi \in \Psi} \varphi.
	\end{align*}
\end{proof}

Note that the set $\Psi$ in Theorem \ref{thm.approximation} is countable and given constructively, let us say $\Psi = \{\psi_1,\psi_2,...\}$. Then the first-order formulas $\psi_1, \psi_1 \land \psi_2,...$ approximate the second-order formula $\exists \mybar{X} \psi$.

\section{Inductive theorem proving}\label{sec.IndProving}

In this section we consider an approach to inductive theorem proving based on tree
 grammars \cite{Eberhard15Inductive} and show how to obtain a central result of \cite{Eberhard15Inductive} from the Horn fixed-point theorem shown in Section \ref{sec.afpthm}. The aim of this approach is to generate a proof of a universal statement in two phases. In the first phase proofs of small instances are computed, from which a second-order unification problem is deduced. Each solution of the unification problem is an induction invariant. We will not go into depth about phase one as we are more interested in phase two, in which solutions of the second-order unification problem are computed. We will see that the second-order unification problem is in fact a Horn formula equation, where we use the Horn fixed-point theorem to get solutions.
In~\cite{Eberhard15Inductive}, as a running example, the approach is demonstrated on a proof
 that the head-recursive and the tail-recursive definitions of the factorial function compute
 the same function.
For space-reasons we do not repeat this example here.
The interested reader is referred to~\cite{Eberhard15Inductive}, and in particular
 to Section 7.1 which contains the computation of a solution of the
 Horn formula equation induced by this running example.
This approach has been developed for the natural numbers in~\cite{Eberhard15Inductive}.
However, this is not a fundamental restriction; the approach could easily be extended to
 induction on recursive data types such as lists, trees, etc.

In this chapter we work in a language $\calL$, which contains the constant symbol $0$ and the unary function symbol $s$. We define $\mybar{n} = s^n(0)$.

\begin{definition}
	Let $F_1,...,F_n,G_1,...,G_m$ be formulas. A formula of the form $F_1 \wedge \cdots \wedge F_n \rightarrow G_1 \vee \cdots \vee G_m$ is called a \emph{sequent} and is written as $\Gamma \Rightarrow \Delta$, where $\Gamma = \{F_1,...,F_n\}$ and $\Delta = \{G_1,...,G_m\}$.
\end{definition}

\begin{definition}[\cite{Eberhard15Inductive},Definition 6.1.]\label{inductiveSchematicSip}
	Let $\alpha,\beta,\nu,\gamma$ be variables only occurring where indicated. Let $\Gamma_0(\alpha,\beta),\Gamma_1(\alpha,\nu,\gamma),\Gamma_2(\alpha)$ be multisets of quantifier-free first-order formulas and let $B(\alpha)$ be a quantifier-free formula. Let $t_i(\alpha,\nu,\gamma)$ and $u_j(\alpha)$ be terms for $i \in \{1,...,n\}$ and $j \in \{1,...,m\}$, where $n,m \geq 1$. Let $X$ be a ternary predicate variable. Then the list of the following three sequents is a \emph{schematic simple induction proof (schematic s.i.p.)}:
	\begin{enumerate}
		\item $\Gamma_0(\alpha,\beta) \Rightarrow X(\alpha,0,\beta)$
		\item $\Gamma_1(\alpha,\nu,\gamma), \bigwedge_{1\leq i \leq n}X(\alpha,\nu,t_i(\alpha,\nu,\gamma)) \Rightarrow X(\alpha,s(\nu),\gamma)$
		\item $\Gamma_2(\alpha), \bigwedge_{1\leq j \leq m} X(\alpha,\alpha,u_j(\alpha)) \Rightarrow B(\alpha)$
	\end{enumerate}	
\end{definition}

Note that every schematic s.i.p. defines a  Horn formula equation $\exists X \forall \alpha,\beta,\nu,\gamma ~ \psi$, where $\psi$ is the conjunction of the three sequents. A solution of a schematic s.i.p. $S$ is defined to be a quantifier-free formula $F(x,y,z)$, such that the three sequents of $S$ with $X$ replaced by $F$ are quasi-tautological, which means it is valid in first-order logic with equality. This means exactly $\models \forall \alpha,\beta,\nu,\gamma ~\psi[X\backslash F]$ as we always talk about logic with equality.
For solving a schematic s.i.p., $\Gamma_0,\Gamma_1,\Gamma_2$ may be arbitrary multisets of
 first-order formulas.
In~\cite{Eberhard15Inductive} the aim is to prove a universal statement $\forall \alpha B(\alpha)$.
A solution of an appropriate schematic s.i.p.\ yields a proof of $\forall \alpha B(\alpha)$ as follows.
Of particular interest is the case, when $\Gamma_0,\Gamma_1,\Gamma_2$ are instances of a theory
 $\Gamma$.
In applications this is an arithmetic theory, e.g. Robinson arithmetic.
Assume we have a solution $F$ of a schematic s.i.p. $S$ in that specific case.
Then we can obtain a proof of $\forall \Gamma \Rightarrow \forall \beta ~F(\alpha,0,\beta)$ from the
 first sequent of $S$, where $\forall \Gamma$ is the universal closure of $\Gamma$.
Similarly we obtain $\forall \Gamma, \forall \beta F(\alpha,\nu, \beta) \Rightarrow \forall \beta F(\alpha,s(\nu),\beta)$
 from the second sequent.
As we are interested in proofs with an induction rule, we then are able to deduce
 $\forall \Gamma \Rightarrow \forall \nu,\beta F(\alpha, \nu, \beta)$.
Thus using the third sequent of $S$ yields a proof of $\forall \Gamma \Rightarrow B(\alpha)$ and
 therefore of $\forall \Gamma \Rightarrow \forall \alpha B(\alpha)$.

Next we examine how to get a solution of a schematic s.i.p.

\begin{definition}[\cite{Eberhard15Inductive}, Definition 6.10.]\label{inductiveSchematicSequence}
	Let $S$ be a schematic s.i.p. with premises given as in Definition \ref{inductiveSchematicSip}. By recursion define the following sequence of formulas.
	\begin{align*}
		C_{S,0}(x,z) &:= \bigwedge \Gamma_0(x,z) \\
		C_{S,q+1}(x,z) &:= \bigwedge \Gamma_1(x,\mybar{q},z) \wedge \bigwedge_{1 \leq i \leq n} C_{S,q}(x,t_i(x,\mybar{q},z))
	\end{align*} 
	$C_{S,q}(x,z)$ is called the $q$-th canonical solution of $S$.
\end{definition}

Let $(\sigma_{\Phi_\psi}^l)_{l \in \omega}$  be the sequence of first-order formulas, describing the unfolding of the fixed-point operator, defined in Definition \ref{def.apprSigma} from the $1$-tuple $\Phi_{\psi}$ of first-order formulas $(\varphi)$ defined in Definition \ref{def.HornformulaSplit} from the Horn formula equation $\exists X \psi$.

\begin{lemma}\label{inductiveEquivalenceVarphi}
	Let $S$ be a schematic s.i.p. and let $\exists X \psi$ be the formula equation defined from $S$. Then for all $q \in \Nat$ it holds 
	\begin{align*}
		C_{S,q}(x,z) \equiv \sigma_{\Phi_{\psi}}^{q+1}(x,\mybar{q},z).
	\end{align*}
\end{lemma}
\begin{proof}
	We show the equivalence by induction on $q$.
	The definition of $(\sigma_{\Phi}^l)_{l \in \omega}$ is $\sigma_{\Phi}^0 \equiv \bot$ and for $l \in \omega$
	\begin{align*}
		\sigma_{\Phi}^{l+1}(x,y,z) \equiv \varphi(\sigma_{\Phi}^l,x,y,z) \equiv& \exists \alpha,\beta, \nu, \gamma \bigg(\Big( \bigwedge \Gamma_0(\alpha,\beta) \wedge x = \alpha \wedge y = 0 \wedge z = \beta\Big)\bigg. \\
		\vee \big.\Big(\bigwedge \Gamma_1&(\alpha,\nu,\gamma) \wedge \bigwedge_{1 \leq i \leq n} \sigma_{\Phi}^l(\alpha,\nu,t_i(\alpha,\nu,\gamma)) \wedge x = \alpha \wedge y = s(\nu) \wedge z = \gamma\Big)\bigg).
	\end{align*}
	In particular it holds that $\sigma_{\Phi}^1(x,0,z) \equiv \bigwedge \Gamma_0(x,z) \equiv  C_{S,0}(x,z)$. For $q \geq 0$ we have 
	\begin{align*}
		\sigma_{\Phi}^{q+2}(x,\mybar{q+1},z) \equiv& \bigwedge \Gamma_1(x,\mybar{q},z) \wedge \bigwedge_{1 \leq i \leq n} \sigma_{\Phi}^{l+1}(x,\mybar{q},t_i(x,\mybar{q},z))\\
		\equiv&  \bigwedge \Gamma_1(x,\mybar{q},z) \wedge \bigwedge_{1 \leq i \leq n} C_{S,q}(x,t_i(x,\mybar{q},z)) \equiv~ C_{S,q+1}(x,z).
	\end{align*}		
\end{proof}

Now we want to present Lemma 6.12. from \cite{Eberhard15Inductive} as a direct corollary of Corollary \ref{cor.HornFP}. 

\begin{lemma}[\cite{Eberhard15Inductive}, Lemma 6.12.]
	Let $S$ be a schematic s.i.p. Then for any solution $F$ of $S$ and any $q \in \Nat$, the $q$-th canonical solution $C_{S,q}(x,z)$ logically implies $F(x,\mybar{q},z)$.
\end{lemma}
\begin{proof}
	From the definition of $\sigma_{\Phi}^{\omega}$ it follows that $\models \sigma_{\Phi}^q \rightarrow \sigma_{\Phi}^{\omega}$ for all $q \in \omega$. Theorem \ref{thm.apprLfp} states that $\sigma_{\Phi}^{\omega} \equiv [\lfp_X~ \Phi_{\psi}]$ and thus Lemma \ref{inductiveEquivalenceVarphi} yields $\models  C_{S,q}(x,z) \rightarrow [\lfp_X~ \Phi_{\psi}]$ for all $q \in \omega$. Now Corollary \ref{cor.HornFP} concludes $\models C_{S,q}(x,z) \rightarrow F(x,\mybar{q},z)$ for all $q \in \omega$.
\end{proof}

Thus we see that using the results from Section~\ref{sec.afpthm} shortens the proofs of a result
 from~\cite{Eberhard15Inductive} and explains it as a corollary of our fixed-point theorem.

\section{Related Work}\label{sec.related_work}

Much of the motivation for studying the solutions of formula equations rests on the many
 connections this problem has to a wide variety of other topics in computational logic.
In addition to the connections described in the previous sections, in this section
 we discuss work on various related topics.

\paragraph{Boolean equations and unification.}
Under the name of Boolean equations, solving formula equations is one of the oldest
 and one of the most central problems of logic.
It goes back to the 19th century and was already thoroughly investigated
 in \cite{Schroeder90Vorlesungen}.
See~\cite{Rudeanu74Boolean} for a comprehensive textbook.
Solving Boolean equations is closely related to Boolean unification, a subject of
 thorough study in computer science, see, e.g.,~\cite{Baader98Complexity,Buettner87Embedding,Martin88Unification}
 and~\cite{Martin89Boolean} for a survey.
The generalisation of this problem from propositional to first-order logic has been
 made explicit as early
 as~\cite{Behmann50Aufloesungsproblem,Behmann51Aufloesungsproblem},
 see~\cite{Wernhard17Boolean} for a recent survey.

%\paragraph{Unification}
%TODO: mention context unification and/or second-order unification?

\paragraph{Second-order quantifier elimination.}
Solving a formula equation in first-order logic is closely related to second-order
 quantifier elimination (SOQE), a problem with applications in a variety of areas in
 computer science, e.g., modal logic, databases, and common-sense
 reasoning~\cite{Gabbay08Second}.
A seminal work on second-order quantifier elimination and basis of many modern algorithms
 is Ackermann's~\cite{Ackermann35Untersuchungen}.
The two main modern algorithms for SOQE are the SCAN algorithm and the DLS algorithm.
SCAN has been introduced in~\cite{Gabbay92Quantifier} and is closely related to
 hierarchic superposition~\cite{Bachmair94Refutational}.
Implementations of SCAN can be found in~\cite{Engel96Quantifier,Ohlbach96SCAN}.
It has been used for a range of applications in various areas of
 computational logic~\cite{Goranko03Scan,Gabbay08Second}.
The DLS algorithm has been introduced in~\cite{Doherty97Computing} and extended to
 the $\text{DLS}^*$ algorithm that works with fixed-points
 in~\cite{Doherty98General,Nonnengart98Fixpoint}.
From the perspective of second-order quantifier elimination, our fixed-point theorem can
 be seen as establishing that Horn formula equations have canonical witnesses, called
 ELIM-witnesses in~\cite{Wernhard17Boolean}, in FO[LFP].

\paragraph{Fixed point theorems.}
Fixed-point theorems play a central role for solving equations in many areas of mathematics.
They abound in logic with one of the most famous examples being the recursion theorem in computability
theory which guarantees the existence of a solution of a system of recursion equations by computing a fixed-point.
But also in areas quite remote from logic such constructions can be found, as, for example, in
the use of Banach's fixed-point theorem in the proof of the
Picard-Lindelöf theorem on the unique solvability of ordinary differential equations.
Our use of the fixed-point theorem for solving Horn formula equations in FO[LFP]
 follows exactly the same scheme.

\paragraph{Least fixed point logics.}

First-order logic with a least fixed-point operator, FO[LFP], has been
 studied at least since~\cite{Moschovakis74Elementary}.
It is well-known in finite model theory and computational complexity, most
 notably because of the Immerman-Vardi theorem~\cite{Vardi82Complexity,Immerman86Relational}.
FO[LFP] has already been used for second-order quantifier elimination in the literature, for example
 by Nonnengart and Sza{\l}as in~\cite{Nonnengart98Fixpoint}.
In fact, the key lemma for our fixed-point theorem is a generalisation of a result
 of~\cite{Nonnengart98Fixpoint}, which, in turn, is a generalisation of a result of
 Ackermann's~\cite{Ackermann35Untersuchungen}.

\paragraph{Logic programming.}
In logic programming it is well-known that a set of first-order Horn clauses has a unique
 minimal model, i.e., a satisfying set of ground atoms, and that this set can be obtained
 as fixed-point of an operator defined by the clause set~\cite{Jaffar94Constraint}.
From that point of view, our fixed-point theorem can be understood as, essentially, expressing
 this computation within the logic itself.

\paragraph{Verification.}

An entire workshop series is devoted to applications of constrained Horn clauses in
 verification and synthesis, see, e.g.,~\cite{Hojjat21Proceedings}.
A set of constrained Horn clauses is simply a Horn formula equation in the terminology of this
 paper.
Also in this community the relationship between Horn formula equations and least fixed points
 is well known and has been exploited for practical purposes, e.g., in~\cite{Unno17Automating},
 see also~\cite{Tsukada22Software}.
In contrast to this line of work, our contribution is the formulation of a general,
 theoretical, result based on an explicit fixed point operator that encompasses both,
 simultaneous least fixed points and abstract interpretation.
Abstract interpretation has been used in tools for solving Horn
 clauses~\cite{Hoder11muZ,Kafle16Rahft}.

The work~\cite{Blass87Existential} advocates for the use of existential fixed-point logic
 as a logical foundation of verification.
In particular, it was already shown in~\cite{Blass87Existential} that the weakest precondition and
 the strongest postcondition can be defined without expressivity hypothesis in
 existential fixed-point logic.
We consider this paper as a continuation of this approach which adds to it by
 using the more abstract concept of formula equation and showing how this is useful for a wide
 range of different application in- and outside of verification.

\paragraph{Extensions.}

Various extensions of Horn formula equations have been considered, mostly motivated
 by their applications in verification.
In particular, Horn formula equations have been extended to coinduction
 in~\cite{Basold19Coinduction} and to higher-order logic~\cite{Jochems23Higher,Kobayashi19Temporal,Tsukada20Computability,Unno23Modular,Burn18Higher, Kobayashi18Higher,Kobayashi23Validity}.
Our work, being restricted to first-order Horn formula equations and least fixed points
 does not directly apply to these results and extending our work to these settings
 is left as future work by this paper.

% \cite{Basold19Coinduction} gives a coinductive semantics for Horn clauses -- how does our work relate to this?
%
% \cite{Jochems23Higher} extends the decidable class MSL of first-order Horn clauses to
%  higher-order logic.

\section{Conclusion}

% summary
We have shown a fixed-point theorem for Horn, dual-Horn and linear-Horn formula equations.
These fixed-point theorems apply to an abstract semantics which generalises standard
 semantics by allowing to interpret the least fixed-point operator and the second-order quantifier in a complete lattice.
The central lemma for proving these results is a generalisation of a result from~\cite{Nonnengart98Fixpoint},
 which, in turn, is a generalisation of a result from~\cite{Ackermann35Untersuchungen}.
From the point of view of logic programming, our proof can be understood as  
 expressing the construction of a minimal model of a set of Horn clauses on the object level as a formula
 in first-order logic with least fixed points.

% theoretical understanding of constrained Horn clause solving and software verification
These fixed-point theorems contribute to our theoretical understanding of the logical foundations of
 constrained Horn clause solving and software verification.
As corollary to our fixed point theorem we have obtained the expressibility of the
 weakest precondition and the strongest postcondition, and thus the partial correctness of an imperative
 program in FO[LFP].
We believe that it is fruitful to consider constrained Horn clause solving from the more general
 point of view of solving formula equations.
On the theoretical level this perspective uncovers connections to a number of topics such as
 second-order quantifier elimination and results such as Ackermann's~\cite{Ackermann35Untersuchungen}.
On the practical level it suggests to study the applicability of algorithms such as DLS and SCAN for
 constrained Horn clauses and vice versa, that of algorithms for constrained Horn clause solving for
 applications of second-order quantifier elimination.

% concrete corollaries
Moreover, we have shown that this fixed-point theorem has a number of applications throughout
 computational logic:
as described in Section~\ref{sectionAffineSolutionProblem} it allows to considerably simplify the
 proof of the decidability of affine formula equations given in~\cite{Hetzl20Decidability}.
As shown in Section~\ref{sec.approximation}, it allows a generalisation of a result by
 Ackermann~\cite{Ackermann35Untersuchungen} on second-order quantifier-elimination in a direction different
 from the recent generalisation~\cite{Wernhard17Approximating} of that result.
Furthermore, as shown in Section~\ref{sec.IndProving}, it allows to obtain a result on the generation of
 a proof with induction based on partial information about that proof shown in~\cite{Eberhard15Inductive}
 as straightforward corollary.
Due to the naturalness of both, the class of formula equations and the result, we expect many further
 applications of this fixed-point theorem.
 
% conclusion
In conclusion we believe that Horn formula equations have a central role to play in computational logic
 and that the fixed-point theorem is one of their most important theoretical properties.
The proof of the fixed-point theorem integrates neatly with the existing literature on second-order
 quantifier elimination.
The expressivity of Horn formula equations enables their applicability in a wide range of topics from
 software verification to inductive theorem proving.

{\bf Acknowledgements.} We would like to thank the anonymous reviewers for their thorough
 reading of our submission and for their many comments which have led to considerable
 improvements of this paper.
 We would also like to thank everyone who has provided feedback on the workshop papers~\cite{Hetzl21Fixpoint} and~\cite{Hetzl21Abstract} which were instrumental in
 developing the full results presented in this paper.
 
\bibliographystyle{ACM-Reference-Format}
\bibliography{references}

\end{document}